\documentclass[a4paper,
11pt,openright]{article} % Per avere margini destro e sinistro diversi
\usepackage[english]{babel} % per lingua. TeXnic dà warning per misteriosi motivi...
\usepackage[latin1]{inputenc} % per gli accenti e carattere
\usepackage{amsmath,amssymb,amsthm,amsfonts,amscd,authblk,dsfont,color,cancel,diagbox,enumerate,epsfig,eucal,fancyhdr,latexsym,lmodern,mathrsfs,mathtools,multicol,multirow,musicography,phaistos,skull,simpler-wick,subcaption,tikz-cd,xfrac}
\usepackage[normalem]{ulem} % emph still at disposal no weird subline on bibtek
\usepackage[multiuser]{fixme}
\FXRegisterAuthor{nd}{annd}{Nico} % \ndnote{Comment on the text} - available only if draft mode is on
\FXRegisterAuthor{cvdv}{ancvdv}{Chris}
\usepackage[
final=true,                    % treat as final
plainpages=false,              % whatever that means
pdfpagelabels=true,            % strange options
pdfencoding=auto,              % getting worse
unicode=true,                  % unicode is always nice
hypertexnames=true,            % needs both to get index *and*
naturalnames=true              % crossrefs correct
]{hyperref}
\usepackage[all,cmtip]{xy}
\usetikzlibrary{matrix,calc}
\usepackage[autostyle, english = american]{csquotes}
\MakeOuterQuote{"} % Sistema l'annosa questione delle virgolette

\definecolor{NiColor}{RGB}{77,77,255}
\definecolor{NiColoRed}{RGB}{255,77,77}
\definecolor{NiCitation}{RGB}{0,181,26}

\newtheoremstyle{TheoremStyle}% <name>
{\topsep}% <Space above>
{\topsep}% <Space below>
{\itshape}%{\slshape}% <Body font>
{}% <Indent amount>
{\sc}%{\bf}%{\itshape}% <Theorem head font>
{:}% <Punctuation after theorem head>
{.5em}% <Space after theorem head>
{}% <Theorem head spec (can be left empty, meaning 'normal')>
\makeatletter
\def\@endtheorem{\hfill$\lozenge$} % This way there is a bigskip after Theorem, Remarks, etc,...
\makeatother

\theoremstyle{TheoremStyle}
\newtheorem{theorem}{Theorem}

\newtheorem{proposition}[theorem]{Proposition}
\newtheorem{lemma}[theorem]{Lemma}
\newtheorem{definition}[theorem]{Definition}
\newtheorem{remark}[theorem]{Remark}
\title{Strict deformation quantization and local spin interactions}
\author[a]{\href{mailto:nicolo.drago@unitn.it}{N. Drago}}
\author[b]{\href{mailto:christiaan.vandeven@mathematik.uni-wuerzburg.de}{C. J. F. van de Ven}}
\affil[a]{Dipartimento di Matematica, Universit\`a di Trento and INFN-TIFPA and INdAM, Via Sommarive 14, I-38123 Povo, Italy}
\affil[b]{Julius Maximilian University of W\"urzburg, Department of Mathematics Chair of Mathematics X (Mathematical Physics), Emil-Fischer-Stra\ss e 31, 97074 W\"urzburg, Germany}

\begin{document}
\maketitle

%\cvdvnote{\christiaannotation{This is the command\\}{\textsf{christiaan\\notation}}}
\begin{abstract}
	We define a strict deformation quantization which is compatible with any Hamiltonian with local spin interaction (\textit{e.g.} the Heisenberg Hamiltonian) for a spin chain.
	This is a generalization of previous results known for mean-field theories.
	The main idea is to study the asymptotic properties of a suitably defined algebra of sequences invariant under the group generated by a cyclic permutation.
	Our point of view is similar to the one adopted by Landsman, Moretti and van de Ven \cite{Landsman_Moretti_vandeVen_2020}, who considered a strict deformation quantization for the case of mean-field theories.
	However, the methods for a local spin interaction are considerably more involved, due to the presence of a strictly smaller symmetry group.
\end{abstract}

\tableofcontents

\section{Introduction}
\label{Sec: Introduction}

In this paper we provide a rigorous $C^*$-algebraic framework for the study of the semi-classical properties of 
any Hamiltonian with local spin interaction for a spin chain.
This covers, for example, the Heisenberg Hamiltonian.
This result is achieved by means of a suitable strict deformation quantization, whose construction is the main result of this paper ---\textit{cf.} Theorem \ref{Thm: deformation quantization of the algebra of gamma sequences}.

Strict deformation quantization originates with Berezin \cite{Berezin_1975} and Bayen et al. \cite{Bayen_Flato_Fronsdal_Lichnerowicz_1978_I,Bayen_Flato_Fronsdal_Lichnerowicz_1978_II} and it is based on the idea of "deforming" a given commutative Poisson algebra representing a classical system into a given non-commutative algebra modelling the associated quantized system.
In Rieffel's approach \cite{Rieffel_1994} the deformed algebras are $C^*$-algebras.

Notably, the "classical-to-quantum" interpretation of a strict deformation quantization is not the unique point of view which can be taken.
In Landsman's approach \cite{Landsman_1998,Landsman_2017} the starting point of a strict deformation quantization is often taken to be a continuous field of $C^*$-algebras.
The latter models an increasingly larger sequence of quantum physical systems, whose limit defines a macroscopic classical theory.
The advantage of this point of view is that it leads to a rigorous notion of the \textit{classical limit} of quantum theories \cite{Landsman_2017}.
This in turn yields a mathematically sound description of several physically interesting \textit{emergent} phenomena, \textit{e.g.} symmetry breaking \cite{Moretti_vandeVen_2021,vandeVen_2020,vandeVen_2021}.
This paper is considering this "micro-to-macro" point of view on strict deformation quantization.

From a technical point of view a \textbf{strict deformation quantization} is defined by the following data:
\begin{enumerate}
	\item
	A commutative Poisson $C^*$-algebra $\mathcal{A}_\infty$, namely a commutative $C^*$-algebra $\mathcal{A}_\infty$ equipped with a Poisson structure $\{\;,\;\}\colon\widetilde{\mathcal{A}}_\infty\times\widetilde{\mathcal{A}}_\infty\to\widetilde{\mathcal{A}}_\infty$ defined on a dense $*$-subalgebra $\widetilde{\mathcal{A}}_\infty\subseteq\mathcal{A}_\infty$ ---\textit{cf.} Section \ref{Subsec: the Poisson structure of the algebra of gamma-sequences}.
	
	\item
	A continuous bundle of $C^*$-algebras \cite{Dixmier_1977} $\prod_{N\in\overline{\mathbb{N}}}\mathcal{A}_N$, where $\overline{\mathbb{N}}=\mathbb{N}\cup\{\infty\}$;
	(Thorough the whole paper we will stick to the case of continuous bundle of $C^*$-algebras over $\overline{\mathbb{N}}$, see \cite{Landsman_2017} for the generic case.)
	\item
	A family of linear maps, called \textbf{quantization maps}, $Q_N\colon\widetilde{\mathcal{A}}_\infty\to\mathcal{A}_N$, $N\in\overline{\mathbb{N}}$, such that
	\begin{enumerate}
		\item\label{Item: quantization maps are Hermitian and define a continuous section}
		$Q_\infty=\operatorname{Id}_{\widetilde{\mathcal{A}}_\infty}$ and $Q_N(a)^*=Q_N(a^*)$ for all $a\in\widetilde{\mathcal{A}}_\infty$.
		Moreover, the assignment
		\begin{align*}
			\overline{\mathbb{N}}\ni N\mapsto Q_N(a)\in\mathcal{A}_N\,,
		\end{align*}
		defines a continuous section of the bundle $\prod_{N\in\overline{\mathbb{N}}}\mathcal{A}_N$.
		\item\label{Item: quantization maps fulfils the DGR condition}
		For all $a,a'\in\widetilde{\mathcal{A}}_\infty$ it holds
		\begin{align}\label{Eq: Dirac-Groenewold-Rieffel condition}
			\lim_{N\to\infty}\|Q_N(\{a,a'\})-iN[Q_N(a),Q_N(a')]\|_{\mathcal{A}_N}=0\,.
		\end{align}
		
		\item\label{Item: quantization maps are strict}
		For all $N\in\mathbb{N}$, $Q_N(\widetilde{\mathcal{A}}_\infty)$ is a dense $*$-subalgebra of $\mathcal{A}_N$.
	\end{enumerate}
\end{enumerate}
The algebra $\mathcal{A}_\infty$ represents the classical (macroscopic) observables of the physical system.
Likewise, the fibers $\mathcal{A}_N$, $N\in\mathbb{N}$, of the bundle $\prod_{N\in\overline{\mathbb{N}}}\mathcal{A}_N$ recollect the quantum observables of the (increasingly larger) quantum system.
	
A relevant example is the strict deformation quantization described in \cite{Landsman_Moretti_vandeVen_2020,Raggio_Werner_1989} for the $C^*$-algebra $[B]^\infty_\pi$ of (equivalence classes of) symmetric sequences --- \textit{cf.} Section \ref{Sec: The algebra of gamma-sequences} for further details.
In this scenario the role of the commutative $C^*$-algebra $\mathcal{A}_\infty$ is played by $[B]^\infty_\pi=C(S(B))$ ---here $B=M_\kappa(\mathbb{C})$, $\kappa\in\mathbb{N}$, while $S(B)$ denotes the states space over $B$.
The continuous bundle of $C^*$-algebras $\prod_{N\in\mathbb{N}}B^N_\pi$ is such that, for $N\in\mathbb{N}$, $B^N_\pi\subseteq B^N$ is the $N$-th symmetric tensor product of $B$.

From a physical point of view the ensuing quantization maps $Q_N$ are of particular interest as they relate to mean-field theories like the Curie-Weiss model \cite[\S 2]{Friedli_Velenik_2017}, for which the interaction between $N$ spin sites is described by
\begin{align}\label{Eq: Curie-Weiss Hamiltonian}
	H_{\textsc{cw},N}:=-\frac{J}{2N}\sum_{j_1+j_2=N-2}^N
	I^{j_1}\otimes\sigma_3\otimes I^{j_2}\otimes\sigma_3
	-h\sum_{j=1}^{N-1} I^{N-1-j}\otimes\sigma_1\otimes I^j\,,
\end{align}
where $\sigma_3,\sigma_1\in M_2(\mathbb{C})$ denote the Pauli's matrices while $I\in M_2(\mathbb{C})$ is the identity matrix and $I^j:=I^{\otimes j}$.
Here $J\in\mathbb{R}$ represents the strength of the spin interaction whereas $h\in\mathbb{R}$ models an external magnetic field acting on the system.
As observed in \cite{Landsman_Moretti_vandeVen_2020} one may recognize that 
\begin{align*}
	H_{\textsc{cw},N}/N=Q_N(h_{\textsc{cw}})+R_N\,,
\end{align*}
where $h_{\textsc{cw}}\in C(S(B))$ while $R_N\in B^N_\pi$ is such that $\|R_N\|_N=O(1/N)$.

The physical interpretation is that $C(S(B))$ is the algebra containing macroscopic observables, \textit{i.e.} observables of an infinite quantum system describing classical thermodynamics as a limit of quantum statistical mechanics.
This has furthermore led to a significant contribution in the study of the classical limit of ground states \cite{Landsman_Moretti_vandeVen_2020,Moretti_vandeVen_2020,Moretti_vandeVen_2021,vandeVen_2020}.
More precisely, in such works a mathematically rigorous description of the limit of ground states $\omega_N$ of $H_{\textsc{cw},N}$ in the regime of large particles $N\to\infty$ is given.
In particular, a classical counterpart $\omega_\infty$ (\textit{i.e.} a probability measure) of the quantum ground state $\omega_N\in S(B^N_\pi)$ is constructed  with the property that  $\omega_\infty(a):=\lim\limits_{N\to\infty}\omega_N(Q_N(a))$ for all $a\in C(S(B))$.
Additionally, this algebraic approach has revealed the existence of several physical \textit{emergent phenomena}, see \cite{vandeVen_2021} for an overview.
These results are consistent with the point of view of \cite{Landsman_2017} ---which is also the one considered in this paper--- for which a quantum theory is pre-existing and the classical limit is computed in a second step, not vice versa.

As characteristic for mean-field models, the Curie-Weiss Hamiltonian describes the energy of a system of $N$ spin sites under the assumption that the interaction is {\em non-local}, namely that every spin site interacts with all other spin site.
This leads to interesting results, but it is ultimately an approximation as one would rather expect each spin site to interact with finitely many neighbouring spin sites.
An exemplary model based on such a local interaction is the celebrated quantum Heisenberg Hamiltonian (for a spin chain) \cite[\S 6.2]{Bratteli_Robinson_1997}
\begin{align}\label{Eq: Heisenberg Hamiltonian}
	H_{\textsc{He},N}:=
	-\sum_{j=0}^{N-1}I^{N-2-j}\otimes\sum_{p,q=1}^3 J^{pq}\sigma_p\otimes\sigma_q\otimes I^{j}
	-\sum_{j=0}^{N-1} I^{N-1-j}\otimes \sum_{p=1}^3 h^p\sigma_p
	\otimes I^j\,,
\end{align}
where $J^{pq}$ is the symmetric matrix describing the spin interaction while $h^p$ are the components of an external magnetic field ---here for $j=N-1$ the contribution in the first sum reads $\sum\limits_{p,q=1}^3J^{pq}\sigma_q\otimes I^{N-2}\otimes\sigma_p$.
For this model the interaction is restricted to two neighbouring sites.

Similarly to what happens with mean-field models one may wonder whether there exists a strict deformation quantization of a suitable $C^*$-algebra such that \begin{align*}
	H_{\textsc{He},N}/N=Q_N(h_{\textsc{He}})+O(1/N)\,.
\end{align*}
The purpose of this paper is to prove that this is in fact the case, \textit{cf.} Theorem \ref{Thm: deformation quantization of the algebra of gamma sequences}.
In \cite{Murro_vandeVen_2022} a different (though similar in spirit) point of view is taken, and a strict deformation is considered such that $H_{\textsc{He},N}=Q_\kappa(h_{\textsc{He},N})+O(1/\kappa)$ where the semi-classical parameter $\kappa$ corresponds to the increasing dimension of the single site algebra $B=M_\kappa(\mathbb{C})$ for a \textit{fixed} number $N$ of lattice sites.
In contrast, this paper deals with an arbitrary but fixed dimension $\kappa\in\mathbb{N}$ considering instead the increasing number $N$ of spin sites as the semi-classical parameter.

Our result is particularly relevant because it provides an excellent basis for studying the classical limit of local quantum spin systems.
Similarly to the case of mean-field theories \cite{Landsman_Moretti_vandeVen_2020,Moretti_vandeVen_2021,vandeVen_2020,vandeVen_2021}, one may now consider a rigorous $C^*$-algebraic formalization of the limit of ground states or Gibbs states \cite{Drago_vandeVen_2023,Friedli_Velenik_2017}.
The latter can be used for the study of spontaneous symmetry breaking and phase transitions in realistic models such as the Heisenberg model.

From a technical point of view, the methods of this paper profit of those of \cite{Landsman_Moretti_vandeVen_2020,Raggio_Werner_1989} for mean-field models.
Nevertheless the results obtained therein do not apply straight away to our case.
As a matter of fact the strict deformation quantization for mean-field models (like the Curie-Weiss Hamiltonian) profits of:
\begin{enumerate}
	\item
	A large symmetry group, that is, mean-field models are symmetric under the permutation group $\mathfrak{S}_N$ of all $N$ spin sites.
	This leads to a high symmetry property which can be exploited in several steps of the construction, \textit{cf.} \cite{Landsman_Moretti_vandeVen_2020,Raggio_Werner_1989}.
	\item
	A fairly explicit description of the classical algebra $[B]^\infty_\pi=C(S(B))$.
	One may define $[B]^\infty_\pi$ in terms of equivalence classes of "symmetric sequences" ---\textit{cf.} Remark \ref{Rmk: gamma-sequences are bounded; comparison with symmetric-sequences}--- but the description in terms of $C(S(B))$ simplifies the discussion, \textit{e.g.} it allows to identify a Poisson structure in a rather direct way.
\end{enumerate}
Contrary to this case, local quantum spin Hamiltonians (\textit{e.g.} the Heisenberg model defined in \eqref{Eq: Heisenberg Hamiltonian}) are invariant under the strictly smaller subgroup generated by a fixed cyclic permutation of $N$ objects.
This spoils the possibility of applying the arguments of \cite{Landsman_Moretti_vandeVen_2020,Raggio_Werner_1989}.
The latter have to be reconsidered to take into account the smaller symmetry group.
Moreover, the classical algebra $[B]^\infty_\gamma$ for such models does not have a "simple" explicit description.
As a matter of fact, $[B]^\infty_\gamma$ is defined as the $C^*$-algebra generated by (equivalence classes of) "$\gamma$-sequences" ---\textit{cf.} Definition \ref{Def: C*-algebra of gamma-sequences}.
Nevertheless it is still possible to prove all properties of $[B]^\infty_\gamma$ relevant for the discussion of its strict deformation quantization.

The paper is structured as follows.
In Section \ref{Sec: The algebra of gamma-sequences} we introduce the notion of "$\gamma$-sequences" ---\textit{cf.} Definition \ref{Def: gamma-sequence}--- and discuss their properties.
The main result in this section is the proof that the $C^*$-algebra $[B]^\infty_\gamma$ generated by (equivalence classes of) $\gamma$-sequences is a commutative $C^*$-algebra.
The latter will play the role of the classical algebra $\mathcal{A}_\infty$ for which we will present a strict deformation quantization.

In Section \ref{Sec: strict deformation quantization of gamma-sequences} we state and prove the main theorem of this paper, which provides a strict deformation quantization of the commutative $C^*$-algebra $[B]^\infty_\gamma$.
To this avail, Section \ref{Sec: The continuous bundle of Cstar-algebras associated with gamma-sequences} is devoted to prove Proposition \ref{Prop: continuous bundle of Cstar algebra for gamma-sequences} which provides the continuous bundle of $C^*$-algebra $[B]_\gamma$ needed in the formulation of Theorem \ref{Thm: deformation quantization of the algebra of gamma sequences}.
The main technical hurdle of this section is to prove that, given a $\gamma$-sequence $(a_N)_N$, the sequence of the norms $(\|a_N\|_N)_N$ is convergent.
While this is straightforward for symmetric sequences (\textit{i.e.} those used when dealing with mean-field models) for $\gamma$-sequences this is non-trivial and has to be discussed carefully.
Sections \ref{Subsec: canonical representative of dense elements in the algebra of gamma-sequence}-\ref{Subsec: the Poisson structure of the algebra of gamma-sequences} discuss further relevant properties of (equivalence classes of) $\gamma$-sequences as well as the Poisson structure on the $C^*$-algebra $[B]^\infty_\gamma$.
Eventually Theorem \ref{Thm: deformation quantization of the algebra of gamma sequences} is proved by recollecting all results from the previous sections.

For the sake of clarity the following theorem recollects in a concise fashion the content of the main Theorem \ref{Thm: deformation quantization of the algebra of gamma sequences} together with the other relevant results of the paper.

\begin{theorem}[main results]
    The algebra $[B]^\infty_\gamma:=\overline{[\dot{B}]^\infty_\gamma}$ of equivalence classes of $\gamma$-sequences ---\textit{cf.} Definitions \ref{Def: gamma-sequence}-\ref{Def: C*-algebra of gamma-sequences}--- is a commutative $C^*$-algebra which is also endowed with a Poisson structure $\{\;,\;\}_\gamma$ ---\textit{cf.} Propositions \ref{Prop: the C* algebra of gamma-sequences is commutative}-\ref{Prop: Poisson structure on algebra of gamma-sequences}.
    
    Moreover, the data $[B]^\infty_\gamma$ and $B^N_\gamma:=\overline{\gamma}_N(B^N)$ ---\textit{cf.} Equation \eqref{Eq: gamma invariant subalgebra}--- define a continuous bundle $[B]_\gamma$ of $C^*$-algebras ---\textit{cf.} Proposition \ref{Prop: continuous bundle of Cstar algebra for gamma-sequences}.
	
	Finally, there exists a family of quantization maps $Q_N\colon [\dot{B}]^\infty_\gamma\to [B]^N_\gamma$, $N\in\overline{\mathbb{N}}$ such that the data $[B]^\infty_\gamma,[B]_\gamma,\{Q_N\}_{N\in\overline{\mathbb{N}}}$ define a strict deformation quantization ---\textit{cf.} Theorem \ref{Thm: deformation quantization of the algebra of gamma sequences}.
\end{theorem}

\paragraph{Acknowledgements.}
We are indebted with V. Moretti for countless helpful discussions on this project and to K. Landsman for his precious feedback.
N.D. is grateful to M. Dippell for a brief, yet effective, discussion on the proof of Lemma \ref{Lem: inequality for convergence of a sequence}.
C. J. F. van de Ven is supported by a postdoctoral fellowship granted by the Alexander von Humboldt Foundation (Germany).

\paragraph{Data availability statement.}
Data sharing is not applicable to this article as no new data were created or analysed in this study.

\paragraph{Conflict of interest statement.}
The authors certify that they have no affiliations with or involvement in any
organization or entity with any financial interest or non-financial interest in
the subject matter discussed in this manuscript.

\section{The algebra of $\gamma$-sequences}
\label{Sec: The algebra of gamma-sequences}

\subsection{Definition of $\gamma$-sequences}
\label{Subsec: definition of gamma-sequences}

In this section we will introduce $\gamma$-sequences and discuss their properties.

To fix some notations, let $\kappa\in\mathbb{N}$ and set $B:=M_\kappa(\mathbb{C})$.
For the sake of simplicity we shall denote by $B^N:=B^{\otimes N}$, where $N\in\mathbb{N}$, with the convention that $B^0=\mathbb{C}$.
The state space over $B^N$ will be denoted by $S(B^N)$: Given $\eta\in S(B)$ we set $\eta^N:=\eta^{\otimes N}\in S(B^N)$.
Whenever needed we will denote $\overline{\mathbb{N}}=\mathbb{N}\cup\{\infty\}$.

Following \cite{Landsman_Moretti_vandeVen_2020} we denote by $I,b_1,\ldots,b_{\kappa^2-1}$, $I\in B$ being the identity matrix, a basis of $B$ (as a $\mathbb{R}$-vector space) abiding by the requirements
\begin{align}\label{Eq: Mk-basis properties}
	\operatorname{tr}(b_j)=0\,,
	\qquad
	b_j^*=b_j\,,
	\qquad
	[b_j,b_\ell]=ic_{j\ell}^{\phantom{j\ell}m}b_m\,,
	\qquad\forall j,\ell=1,\ldots,\kappa^2-1\,.
\end{align}
where $c_{j\ell}^{\phantom{j\ell}m}$ denotes the structure constants of $\mathfrak{su}(\kappa)$.
In the particular case $\kappa=2$ we may choose $b_j=\sigma_j/2$ while $c_{j\ell}^{\phantom{j\ell}m}=\varepsilon_{j\ell s}\delta^{sm}$, $\varepsilon_{ijk}$ being the Levi-Civita symbol.
We will denote by $\widetilde{B}$ the vector space generated by $\{b_j\}_{j=1}^{\kappa^2-1}$. The latter corresponds to the $\ker\tau$, being $\tau\colon B\to\mathbb{C}$ the normalized trace defined by $\tau(a):=\operatorname{tr}(a)/\kappa$.

We then consider the linear operator (left-shift operator) $\gamma_N\colon B^N\to B^N$ uniquely defined by continuous and linear extension of the following map defined on elementary tensors\footnote{In the forthcoming discussion we will use the notation $a_N$ to denote an element $a_N\in B^N$. When we will need to use a subindex without necessarily stating the degree of the element we will use the notation $a_{(k)}$ so that $a_{(k)}\in B^{M(k)}$, $M(k)\in\mathbb{N}$.}
\begin{align}\label{Eq: gamma-operator}
	\gamma_N(a_{(1)}\otimes\ldots\otimes a_{(N)}):=a_{(2)}\otimes \ldots \otimes a_{(N)}\otimes a_{(1)}
	\qquad a_{(1)},\ldots, a_{(N)}\in B\,.
\end{align}
The operator $\gamma_N$ is an algebra endomorphism, moreover, $\gamma_N^N=\operatorname{Id}_B$, $\operatorname{Id}_B\colon B\to B$ being the identity operator.
We denote by $\overline{\gamma}_N\colon B^N\to B^N$ the \textbf{averaged $\gamma_N$ operator}, defined by
\begin{align}\label{Eq: averaged gamma-operator}
	\overline{\gamma}_N:=\frac{1}{N}\sum_{j=0}^{N-1}\gamma_N^j\,.
\end{align}
Clearly $\gamma_N\circ\overline{\gamma}_N=\overline{\gamma}_N=\overline{\gamma}_N\circ \gamma_N$: We denote by
\begin{align}\label{Eq: gamma invariant subalgebra}
	B^N_\gamma:=\overline{\gamma}_N(B^N)\,,
\end{align}
the $C^*$-subalgebra of $B^N$ made by $\gamma_N$-invariant elements.

Through this paper we will mostly consider sequences $(a_N)_N=(a_N)_{N\in\mathbb{N}}$ with $a_N\in B^N$ for all $N\in\mathbb{N}$.
A sequence $(a_N)_{N\geq K}$ will be implicitly extended to $(a_N')_{N\in\mathbb{N}}$ where $a_N'=a_N$ for $N\geq K$ and $a_N'=0$ for $N<K$.

\begin{definition}\label{Def: gamma-sequence}
	A sequence $(a_N)_N$ is called \textbf{$\gamma$-sequence} if there exists $M\in\mathbb{N}$ and $a_M\in B^M$ such that
	\begin{align}\label{Eq: gamma-sequence property}
		a_N=\overline{\gamma}^M_Na_M
		:=\begin{cases*}
			\overline{\gamma}_N(I^{N-M}\otimes a_M)
			\qquad N\geq M
			\\
			0\qquad N<M
		\end{cases*}\,.
	\end{align}
	where $I\in B$ denotes the identity of $B$ and $I^N=I^{\otimes N}$.
\end{definition}

\begin{remark}\label{Rmk: gamma-sequences are bounded; comparison with symmetric-sequences}
	\noindent
	\begin{enumerate}[(i)]
		\item\label{Item: gamma-sequences are bounded}
		For fixed $N,M\in\mathbb{N}$, $N\geq M$, $\overline{\gamma}^M_N\colon B^M\to B^N_\gamma$ is a linear operator with operator norm smaller than 1.
		This implies that, $(\overline{\gamma}^M_Na_M)_N$ is bounded with
		\begin{align*}
			\|(\overline{\gamma}^M_Na_M)_N\|_\infty
			:=\sup_{N\in\mathbb{N}}\|\overline{\gamma}^M_Na_M\|_N
			\leq\|a_M\|_M\,,
		\end{align*}
		where $\|\;\|_M$ denotes the norm on $B^M$.
		
		\item\label{Item: comparison with symmetric-sequences} 
		It is worth comparing our construction with the one presented in the literature \cite{Landsman_2017,Landsman_Moretti_vandeVen_2020,Raggio_Werner_1989}, based on symmetric sequences.
		We stress that the latter are exploited to deal with the Curie-Weiss Hamiltonian ---or more generally with mean-field theories \cite[\S 10]{Landsman_2017}--- which prescribe a non-local interaction between spin sites.
		On the other hand we are interested in models compatible with Hamiltonian describing a local interaction between spin sites ---\textit{e.g.} the Heisenberg Hamiltonian, \textit{cf.} Remark \ref{Rmk: reconstruction of starting point for gamma-sequence; Heisenberg and CW Hamiltonians; asymptotic consistency for gamma-sequences}.
		To describe the non-local interaction algebraically 
		 one considers the symmetrization operator $S_N\colon B^N\to B^N$ defined by continuous and linear extension of
		\begin{align*}
			S_N(a_{(1)}\otimes\ldots\otimes a_{(N)}):=\frac{1}{N!}\sum_{\varsigma\in\mathfrak{S}_N}
			a_{(\varsigma(1))}\otimes\ldots\otimes a_{(\varsigma(N))}
			\qquad
			a_{(1)},\ldots,a_{(N)}\in B\,,
		\end{align*}
		where $\mathfrak{S}_N$ is the set of permutation of $N$ objects \cite{Landsman_2017,Raggio_Werner_1989}.
		Considering the $C^*$-subalgebra $B_\pi^N:=S_NB^N\subset B^N$ one then defines a \textbf{symmetric-sequence} (shortly, $\pi$-sequence) to be a sequence $(a_N)_N$ such that there exists $M\in\mathbb{N}$ and $a_M\in B^M_\pi$ fulfilling
		\begin{align*}
			(a_N)_N=(\pi^M_Na_M)_N:=
			(S_N(I^{N-M}\otimes a_M))_{N\geq M}\,.
		\end{align*}
		One immediately sees the relation with Definition \ref{Def: gamma-sequence}: Actually a $\gamma$-sequence is defined in a way similar to $\pi$-sequences but averaging over a strictly smaller subgroup of $\mathfrak{S}_N$.
		In fact $\gamma$-sequences and $\pi$-sequences share many similar properties, although $\pi$-sequences are generally speaking better behaved.
	\end{enumerate}
\end{remark}

\subsection{Asymptotic properties of $\gamma$-sequences}
\label{Subsec: asymptotic properties of gamma-sequences}

In what follows we will be mainly interested in the asymptotic behaviour as $N\to\infty$ of the sequences under investigations.
For this reason, following \cite{Raggio_Werner_1989}, we introduce the \textbf{$\sim$-equivalence relation}
\begin{align}\label{Eq: equivalence relation among sequences}
	(a_N)_N\sim(b_N)_N\Longleftrightarrow
	\lim_{N\to\infty}\|a_N-b_N\|_N=0\,.
\end{align}
For a given sequence $(a_N)_N$ we will denoted by $[a_N]_N:=[(a_N)_N]$ the corresponding equivalence class with respect to \eqref{Eq: equivalence relation among sequences}.
The $\sim$-equivalence relation \eqref{Eq: equivalence relation among sequences} has a nice interplay with the \textbf{full $C^*$-product} $\prod_{N\in\mathbb{N}}B^N$ defined by
\begin{align}\label{Eq: full C*-product}
	\prod_{N\in\mathbb{N}}B^N:=\{(a_N)_N\,|\,(\|a_N\|_N)_N\in\ell^\infty(\mathbb{N})\}\,.
\end{align}
As it is well-known \cite{Blackadar_2006} $\prod_{N\in\mathbb{N}}B^N$ is a $C^*$-algebra with respect to sup norm $\|(a_N)_N\|_\infty:=\sup\limits_{N\in\mathbb{N}}\|a_N\|_N$.
Moreover, the \textbf{direct $C^*$-sum}
\begin{align}\label{Eq: direct C*-sum}
	\bigoplus_{N\in\mathbb{N}}B^N:=\{(a_N)_N\in\prod_{N\in\mathbb{N}}B^N\,|\,
	\lim_{N\to\infty}\|a_N\|_N=0\}\,,
\end{align}
is a closed two-sided ideal in $\prod_{N\in\mathbb{N}}B^N$ and thus we may consider the quotient
\begin{align}\label{Eq: quotient algebra}
	[B]_\sim:=\prod_{N\in\mathbb{N}}B^N/\bigoplus_{N\in\mathbb{N}}B^N\,,
\end{align}
which is nothing but the space of $\sim$-equivalence classes $[a_N]_N$ for bounded sequences $(a_N)_N$.
Importantly, $[B]_\sim$ is a $C^*$-algebra with norm
\begin{align}\label{Eq: norm of the quotient algebra}
	\|[a_N]_N\|_{[B]_\sim}=\limsup_{N\to\infty}\|a_N\|_N\,.
\end{align}

\begin{remark}\label{Rmk: reconstruction of starting point for gamma-sequence; Heisenberg and CW Hamiltonians; asymptotic consistency for gamma-sequences}
	\noindent
	\begin{enumerate}[(i)]
		\item\label{Item: gamma- and pi- sequences have equivalent classes}
		Since both $\gamma$- and $\pi$-sequences are bounded ---\textit{cf}. Remark \ref{Rmk: gamma-sequences are bounded; comparison with symmetric-sequences}-\ref{Item: gamma-sequences are bounded}--- they lead to well-defined elements $[\overline{\gamma}^M_N a_M]_N$, $[\pi^M_N a_M]_N\in [B]_\sim$.
		One may wonder whether $[\overline{\gamma}^M_Na_M]_N=[0]_N$ for a non-zero $a_M\in B^M$.
		This is in fact possible, but we postpone this discussion to Section \ref{Subsec: canonical representative of dense elements in the algebra of gamma-sequence} where we will prove that, for a given equivalence class $[\overline{\gamma}^M_Na_M]_N$ it is possible to extract a "canonical representative" ---\textit{cf.} Definition \ref{Def: canonical representative of gamma-sequence}--- with the property that $[\overline{\gamma}^M_Na_M]_N=[0]_N$ if and only if the canonical representative vanishes ---\textit{cf.} Proposition \ref{Prop: the canonical representative of a gamma-sequence is unique}.
		
		\item\label{Item: Heisenberg and CW Hamiltonians}
		With reference to Equation \eqref{Eq: Heisenberg Hamiltonian} we have (considering $\kappa=2$)
		\begin{align*}
			\frac{1}{N}H_{\textsc{He},N}
			=\overline{\gamma}^2_N\bigg(
			\sum_{p,q=1}^3 J^{pq}\sigma_p\otimes\sigma_q
			\bigg)
			+\overline{\gamma}^1_N\bigg(
			\sum_{p=1}^3 h^p\sigma_p
			\bigg)\,,
		\end{align*}
		showing the relation between $\gamma$-sequences and the Heisenberg Hamiltonian.
		Similarly, as discussed in \cite{Landsman_Moretti_vandeVen_2020}, Equation \eqref{Eq: Curie-Weiss Hamiltonian} leads to
		\begin{align*}
			\frac{1}{N}H_{\textsc{cw},N}
			=-\pi^2_N(J\sigma_3\otimes\sigma_3)
			+\pi^1_N(h\sigma_1)
			+O(1/N)\,,
		\end{align*}
		showing that $(H_{\textsc{cw},N}/N)_{N\geq 1}$ is equivalent to a $\pi$-sequence.
		
		At this stage it is worth observing that $\gamma$-sequences model an arbitrary Hamiltonian with local spin interaction.
		We say that $H_N\in B^N_\gamma$ is a (translation invariant) \textbf{Hamiltonian with local spin interaction} if and only if
		\begin{align}\label{Eq: Hamiltonian with local spin interaction}
		    H_N=\sum_{|i-j|\leq \ell}\sum_{p,q=1}^3 J^{pq}\sigma_p(i)\sigma_q(j)
		    +\sum_{i=1}^N\sum_{p=1}^3h^p\sigma_p(i)\,,
		\end{align}
		 were $\sigma_p(i)\sigma_q(j)$ is a short notation for $I^{i-1}\otimes\sigma_p\otimes I^{j-i-1}\otimes\sigma_q\otimes I^{N-i-j}$ and similarly $\sigma_p(i)=I^{i-1}\otimes\sigma_p\otimes I^{N-i}$.
		 The parameter $\ell\in\mathbb{N}$ determines the number of spin sites which interact with a fixed spin site $i$ ---\textit{e.g.} for the Heisenberg Hamiltonian $\ell=1$.
		 The strength of the interaction and of the external magnetic field is determined by $J^{pq}$, $h^p$.
		 Notice that the latter do not depend on the spin site: This entails that we are considering translation invariant local spin interactions.
		 
		 Any Hamiltonian $H_N$ as per Equation \eqref{Eq: Hamiltonian with local spin interaction} leads to a $\gamma$-sequences as per Definition \eqref{Def: gamma-sequence}.
		 Indeed, we have
		 \begin{align*}
		     H_N/N=\sum_{m=0}^{\ell-1}\overline{\gamma}_N^{2+m}\bigg(
		     \sum_{p,q=1}^3 J^{pq}\sigma_p\otimes I^m \otimes\sigma_q
		     \bigg)
		     +\overline{\gamma}_N^1\bigg(
		     \sum_{p=1}^3h^p\sigma_p\bigg)\,.
		 \end{align*}

		\item\label{Item: asymptotic consistency for gamma-sequences}
		The $\sim$-equivalence relation \eqref{Eq: equivalence relation among sequences} provides a first example showing the different behaviour of $\gamma$-sequences with respect to $\pi$-sequences.
		To this avail, let $a_M\in B^M_\pi$ and let us consider the $\pi$-sequence $(\pi^M_Na_M)_N$.
		By direct inspection one immediately sees that, for all $N'\geq N\geq M$
		\begin{align*}
			\pi_{N'}^N\pi_N^Ma_M
			=S_{N'}\bigg[I^{N'-N}\otimes S_N(I^{N-M}\otimes a_M)\bigg]
%			=S_N\bigg[I^{N'-N}\otimes I^{N-M}\otimes a_M\bigg]
			=\pi^M_{N'}a_M\,,
		\end{align*}
		which shows that the family of maps $\pi^M_N\colon B^M_\pi\to B^N_\pi$ is "consistent", namely $\pi^N_{N'}\circ \pi^M_N=\pi^M_N$.
		The same property does not apply for $\gamma$-sequences, but it holds only asymptotically.
		Indeed for $a_M\in B^M$ one has, for $N\geq M$,
		\begin{align*}
			\overline{\gamma}^M_Na_M
			&=\frac{1}{N}\sum_{j=0}^{N-M}
			\gamma_N^j(I^{N-M}\otimes a_M)
			+R_N
			\\&
			=\frac{1}{N}\sum_{j=0}^{N-M}
			I^{N-M-j}\otimes a_M\otimes I^j
			+R_N\,,
			\qquad\|R_N\|_N\leq \frac{M-1}{N}\|a_M\|_M\,.
		\end{align*}
		It then follows that
		\begin{align*}
			\overline{\gamma}^N_{N'}\overline{\gamma}^M_Na_M
			=\overline{\gamma}^N_{N'}\bigg[
			\frac{1}{N}\sum_{j=0}^{N-M}
			I^{N-M-j}\otimes a_M\otimes I^j
			\bigg]
			+\overline{\gamma}^N_{N'}R_N
			=\overline{\gamma}^M_{N'}a_M
			+R'_{N'}\,,
		\end{align*}
		where we used the $\gamma$-invariance while
		\begin{align*}
			\|R'_{N'}\|_{N'}=
			\bigg\|\overline{\gamma}^N_{N'}R_N
			-\frac{M-1}{N}\overline{\gamma}^M_{N'}a_M\bigg\|_{N'}
			\leq\|R_N\|_N+\frac{M-1}{N}\|a_M\|_M
			=O(1/N)\,.
		\end{align*}
		This shows that, although $\overline{\gamma}^N_{N'}\circ\overline{\gamma}^M_N\neq\overline{\gamma}^M_N$ one still has
		\begin{align}\label{Eq: gamma sequences asymptotic consistency}
			\lim_{N\to\infty}
			\|[\overline{\gamma}^N_{N'}\overline{\gamma}^M_Na_M]_{N'}
			-[\overline{\gamma}^M_{N'}a_M]_{N'}\|_{B_\sim}=0\,.
		\end{align}
		As we shall see, naively speaking most the results obtained for $\pi$-sequences holds true also for $\gamma$-sequences but only asymptotically ---in the sense of relation \eqref{Eq: equivalence relation among sequences}.
	\end{enumerate}
\end{remark}

\subsection{The algebra generated by $\gamma$-sequences}
\label{Subsec: the algebra generated by gamma-sequences}

In what follows we will consider the $*$-algebra $\dot{B}^\infty_\gamma\subset\prod_{N\in\mathbb{N}}B^N$ generated by $\gamma$-sequences together with its projection $[\dot{B}]^\infty_\gamma\subset [B]_\sim$.
As we will see, the latter algebra enjoys remarkable properties, in particular, it can be completed to a commutative $C^*$-algebra $[B]^\infty_\gamma$.

\begin{definition}\label{Def: C*-algebra of gamma-sequences}
	Let $\dot{B}^\infty_\gamma$ be the $*$-algebra generated by $\gamma$-sequences ---\textit{cf.} Definition \ref{Def: gamma-sequence}.
	We denote by $[\dot{B}]^\infty_\gamma\subset [B]_\sim$ the projection of $\dot{B}^\infty_\gamma$ in $[B]_\sim$, that is, $[\dot{B}]^\infty_\gamma$ is the $*$-algebra generated by equivalence classes of $\gamma$-sequences.
	Thus, $[a_N]_N\in[\dot{B}]^\infty_\gamma$ if and only if
%	\begin{align*}
%		[a_N]_N=\sum_{\substack{\ell,M_1,\ldots, M_\ell\\k_1,\ldots,k_\ell}}c^{k_1\ldots k_\ell}_{M_1\ldots M_\ell}[\overline{\gamma}_N^{M_1}(a_{M_1,k_1})\cdots\overline{\gamma}^{M_\ell}_N(a_{M_\ell,k_\ell})]_N\,,
%	\end{align*}
	\begin{align*}
		[a_N]_N=\sum_{\substack{\ell,k_1,\ldots,k_\ell}}c^{k_1\ldots k_\ell}
		[\overline{\gamma}_N^{M(k_1)}(a_{(k_1)})\cdots\overline{\gamma}^{M(k_\ell)}_N(a_{(k_\ell)})]_N\,,
	\end{align*}
	where $a_{(k_j)}\in B^{M(k_j)}$ while the sum over $\ell,k_1,\ldots,k_\ell$ is finite.
%	where $a_{M_j,k_j}\in B^{M_j}$ while the sum over $\ell,M_1,\ldots,M_\ell$ is finite.
	We denote by $[B]^\infty_\gamma:=\overline{[\dot{B}]^\infty_\gamma}$ the closure of $[\dot{B}]^\infty_\gamma$ in $[B]_\sim$, that is, the $C^*$-algebra generated by equivalence classes of $\gamma$-sequences.
	To wit, an equivalence class $[a_N]_N$ belongs to $[B]^\infty_\gamma$ if and only if for all $\varepsilon>0$ there exists $N_\varepsilon\in\mathbb{N}$ and $[a'_N]_N\in[\dot{B}]^\infty_\gamma$ such that $\|a_N-a'_N\|_N<\varepsilon$ for all $N\geq N_\varepsilon$.
\end{definition}

\begin{proposition}\label{Prop: the C* algebra of gamma-sequences is commutative}
	Let $a_{M_1},\ldots,a_{M_\ell}$, $\ell\in\mathbb{N}$, be such that $a_{M_j}\in B^{M_j}$, $j=1,\ldots,\ell$.
	Then:
	\begin{align}\label{Eq: equivalence class of product of gamma sequences}
		[\overline{\gamma}^{M_1}_N(a_{M_1})\cdots\overline{\gamma}^{M_\ell}_N(a_{M_\ell})]_N=
		\bigg[\overline{\gamma}_N\bigg(\frac{1}{N^{\ell-1}}\sum_{|j|_\ell=N-|M|_\ell}
		I^{j_1}\otimes
		\wick{
			\c1{a}_{M_1}\otimes
			\ldots\otimes I^{j_\ell}\otimes
			\c1{a}_{M_\ell}
		}
		\bigg)\bigg]_N\,,
	\end{align}
	where $|j|_\ell:=j_1+\ldots+j_\ell$, $|M|_\ell=M_1+\ldots +M_\ell$ is a short notation while
	\begin{align}\label{Eq: total weighted symmetrization}
		I^{j_1}\otimes\wick{\c1 a_{M_1}\otimes\ldots\otimes I^{j_\ell}\otimes\c1{a}_{M_\ell}}
		:=\frac{1}{\ell}
		\sum_{\varsigma\in\mathfrak{S}_\ell}
		I^{j_1}\otimes
		a_{M_{\varsigma(1)}}\otimes\ldots\otimes I^{j_\ell}\otimes
		{a}_{M_{\varsigma(\ell)}}\,,
	\end{align}
	denotes "total weighted symmetrization" over the factor $a_{M_1},\ldots, a_{M_\ell}$
	\footnote{For example $I^{j_1}\otimes\wick{\c1 a_{M_1}\otimes I^{j_2}\otimes\c1{a}_{M_2}}=(I^{j_1}\otimes a_{M_1}\otimes I^{j_2}\otimes a_{M_2}+I^{j_1}\otimes a_{M_2}\otimes I^{j_2}\otimes a_{M_1})/2$.}.
	In particular $[B]^\infty_\gamma$ is a commutative $C^*$-subalgebra of $[B]_\sim$.
\end{proposition}

\begin{proof}
	\noindent
	Let $a_{M_1}\in B^{M_1},\ldots,a_{M_\ell}\in B^{M_\ell}$, $\ell\in\mathbb{N}$.
	We will prove that, for $N$ large enough,
	\begin{multline}\label{Eq: product of gamma sequences - estimate of remainder}
		\overline{\gamma}^{M_1}_N(a_{M_1})\cdots\overline{\gamma}^{M_\ell}_N(a_{M_\ell})=
		\overline{\gamma}_N\bigg(\frac{1}{N^{\ell-1}}\sum_{|j|_\ell=N-|M|_\ell}
		I^{j_1}\otimes
		\wick{
			\c1 a_{M_1}\otimes
			\ldots\otimes I^{j_\ell}\otimes
			\c1 a_{M_\ell}
		}\bigg)
		+R_N\,,
	\end{multline}
	where $R_N\in B^N$ is such that $\|R_N\|_N=O(1/N)$.
	On account of \eqref{Eq: equivalence relation among sequences} this implies Equation \eqref{Eq: equivalence class of product of gamma sequences}.
	
	We proceed by induction over $\ell\in\mathbb{N}$.
	For $\ell=1$ the right-hand side of Equation \eqref{Eq: product of gamma sequences - estimate of remainder} reduces to $\overline{\gamma}^M_Na_M+R_N$ so that we may choose $R_N=0$.
	For $\ell=2$ we find, for $N$ large enough (say, $N\geq 2(M_1+M_2)$),
	\begin{align*}
		\overline{\gamma}^{M_1}_N(a_{M_1})\overline{\gamma}^{M_2}_N(a_{M_2})
		&=\overline{\gamma}_N\bigg(
		\left(I^{N-M_1}\otimes a_{M_1}\right)
		\overline{\gamma}_N(I^{N-M_2}\otimes a_{M_2})
		\bigg)
		\\&=\overline{\gamma}_N\bigg(
		\left(I^{N-M_1}\otimes a_{M_1}\right)
		\frac{1}{N}\sum_{j=0}^{N-1}\gamma_N^j(I^{N-M_2}\otimes a_{M_2})
		\bigg)
		\\&=\overline{\gamma}_N\bigg(
		\left(I^{N-M_1}\otimes a_{M_1}\right)
		\frac{1}{N}\sum_{j=M_1}^{N-M_2-1}\gamma_N^j(I^{N-M_2}\otimes a_{M_2})
		\bigg)
		+R_N
		\\&=\overline{\gamma}_N\bigg(
		\frac{1}{N}\sum_{j_1+j_2=N-M_1-M_2}I^{j_2}\otimes a_{M_2}\otimes I^{j_1}\otimes a_{M_1}
		\bigg)
		+R_N\,,
		\\&=\overline{\gamma}_N\bigg(
		\frac{1}{N}\sum_{j_1+j_2=N-M_1-M_2}I^{j_2}\otimes
		\wick{
			\c1 a_{M_2}
			\otimes I^{j_1}\otimes
			\c1 a_{M_1}}
		\bigg)
		+R_N\,,
	\end{align*}
	where in the last line we used the symmetry of $j_1,j_2$ as well as the $\gamma_N$-invariance of the whole term.
	The remainder $R_N$ coincides with
	\begin{align*}
		\|R_N\|_N
		=\bigg\|
		\overline{\gamma}_N\bigg(
		(I^{N-M_1}\otimes a_{M_1})
		\frac{1}{N}\sum_{\substack{j\in\{0,\ldots,M_1-1\}\\\cup\{N-M_2,\ldots,N-1\}}}\gamma_N^j(I^{N-M_2}\otimes a_{M_2})
		\bigg)
		\bigg\|_N
		\leq\frac{C_{M_1,M_2}}{N}\,,
	\end{align*}
	where $C_{M_1,M_2}>0$ is a constant depending on $a_{M_1},a_{M_2}$.
	Roughly speaking, we removed the values of $j$ for which $a_{M_1}$ and $a_{M_2}$ have "overlapping positions".
	This happens in $M_1+M_2$ cases, whose fraction vanishes as $N\to\infty$.
		
	This proves Equation \eqref{Eq: product of gamma sequences - estimate of remainder} for $\ell=2$.
	Proceeding by induction on $\ell$, we now assume that Equation \eqref{Eq: product of gamma sequences - estimate of remainder} holds for all $\ell'<\ell$ and prove it for $\ell$.
	To this avail we consider, for $N\geq 2|M|_\ell$,
	\begin{multline*}
		\overline{\gamma}^{M_1}_N(a_{M_1})\cdots\overline{\gamma}^{M_\ell}_N(a_{M_\ell})
		\\=\overline{\gamma}_N\bigg(\frac{1}{N^{\ell-2}}\sum_{|j|_{\ell-1}=N-|M|_{\ell-1}}
		I^{j_1}\otimes
		\wick{
			\c1 a_{M_1}\otimes
			\ldots\otimes I^{j_{\ell-1}}\otimes
			\c1 a_{M_{\ell-1}}
		}\bigg)
		\overline{\gamma}^{M_\ell}_N(a_{M_\ell})
		+R_N\overline{\gamma}^{M_\ell}_N(a_{M_\ell})
		\\=\overline{\gamma}_N\bigg(\frac{1}{N^{\ell-2}}\sum_{|j|_{\ell-1}=N-|M|_{\ell-1}}
		(I^{j_1}\otimes
		\wick{
			\c1 a_{M_1}\otimes
			\ldots\otimes I^{j_{\ell-1}}\otimes
			\c1 a_{M_{\ell-1}}
		})\overline{\gamma}^{M_\ell}_N(a_{M_\ell})\bigg)
		+R_N'\,,
	\end{multline*}
	where $\|R_N'\|_N\leq\|R_N\|_N\|a_{M_\ell}\|=O(1/N)$.
	Thus, we focus on
	\begin{multline*}
		\overline{\gamma}_N\bigg(\frac{1}{N^{\ell-2}}
		(I^{j_1}\otimes
		\wick{
			\c1 a_{M_1}\otimes
			\ldots\otimes I^{j_{\ell-1}}\otimes
			\c1 a_{M_{\ell-1}}
		})\overline{\gamma}^{M_\ell}_N(a_{M_\ell})\bigg)
		\\=\overline{\gamma}_N\bigg(\frac{1}{N^{\ell-1}}
		(I^{j_1}\otimes
		\wick{
			\c1 a_{M_1}\otimes
			\ldots\otimes I^{j_{\ell-1}}\otimes
			\c1 a_{M_{\ell-1}}
		})
		\sum_{j_\ell=0}^{N-1}\gamma^{j_\ell}_N(I^{N-M_\ell}\otimes a_{M_\ell})\bigg)\,,
	\end{multline*}
	where $j_1,\ldots, j_{\ell-1}$ are such that $|j|_{\ell-1}=N-|M|_{\ell-1}$.
	We now proceed as in the case $\ell=2$
	by considering only those values $j_\ell$ for which the position of $a_{M_\ell}$ "overlaps" with the ones of $I^{j_1},\ldots, I^{j_\ell}$ and not with those of $a_{M_1},\ldots, a_{M_{\ell-1}}$. 
	Notice that, in focusing only on these $j_\ell$'s we are neglecting a contribution $R_N''$ with $\|R_N''\|_N=O(1/N)$.
	We obtain
	\begin{multline}\label{Eq: product with ell factor - aMell distributes}
		\overline{\gamma}_N\bigg(\frac{1}{N^{\ell-1}}
		(I^{j_1}\otimes
		\wick{
			\c1 a_{M_1}\otimes
			\ldots\otimes I^{j_{\ell-1}}\otimes
			\c1 a_{M_{\ell-1}}
		})
		\sum_{j_\ell=0}^{N-1}\gamma^{j_\ell}_N(I^{N-M_\ell}\otimes a_{M_\ell})\bigg)
		\\=\overline{\gamma}_N\bigg(\frac{1}{N^{\ell-1}}
		\sum_{h_1=0}^{j_1-M_\ell}
		I^{h_1}\otimes a_{M_\ell}
		\otimes I^{j_1-M_\ell-h_1}\otimes
		\wick{
			\c1 a_{M_1}\otimes
			\ldots\otimes I^{j_{\ell-1}}\otimes
			\c1 a_{M_{\ell-1}}
		}\bigg)
		\\+\ldots
		+\overline{\gamma}_N\bigg(\frac{1}{N^{\ell-1}}
		I^{j_1}\otimes
		\wick{
			\c1 a_{M_1}\otimes
			\ldots\otimes
			\sum_{0\leq h_{\ell-1}\leq j_\ell-M_\ell}
			I^{h_{\ell-1}}\otimes a_{M_\ell}
			\otimes I^{j_{\ell-1}-M_\ell-h_{\ell-1}}\otimes
			\c1 a_{M_{\ell-1}}
		}\bigg)
		+R_N''\,,
	\end{multline}
	where $\|R_N''\|_N=O(1/N)$ while the sum over the $h_p$ is empty if $j_p<M_\ell$ ---notice that at least one of these sums is not empty if $N$ is large enough.
	Notice that each of the $\ell-1$ sets of $\ell$ indexes
	\begin{multline*}
		\{h_1,j_1-M_\ell-h_1,j_2,\ldots, j_{\ell-1}\}\,,\, \{j_1,h_2,j_2-M_\ell-h_2,j_3\ldots, j_{\ell-1}\}\,,
		\\\ldots\, \{j_1,\ldots,j_{\ell-2},h_{\ell-1},j_{\ell-1}+M_\ell-h_{\ell-1}\}\,.
	\end{multline*}
	is such that its elements sum to $N-|M|_\ell$.
	Considering now the summation over $j_1,\ldots, j_{\ell-1}$ and using the $\gamma_N$-invariance each subset of indexes provides the same contribution.
        We are lead to
	\begin{multline*}
		\overline{\gamma}_N\bigg(\frac{1}{N^{\ell-2}}\sum_{|j|_{\ell-1}=N-|M|_{\ell-1}}
		(I^{j_1}\otimes
		\wick{
			\c1 a_{M_1}\otimes
			\ldots\otimes I^{j_{\ell-1}}\otimes
			\c1 a_{M_{\ell-1}}
		})\overline{\gamma}^{M_\ell}_N(a_{M_\ell})\bigg)
		\\=\overline{\gamma}_N\bigg(\frac{\ell-1}{N^{\ell-1}}\sum_{|j|_\ell=N-|M|_\ell}
		I^{j_1}\otimes
		\wick{
			\c1 a_{M_1}\otimes
			\ldots\otimes I^{j_{\ell-1}}\otimes
			\c1 a_{M_{\ell-1}}
		}\otimes I^{j_\ell}\otimes a_{M_\ell}\bigg)
		+R_N''\,,
		\\=\overline{\gamma}_N\bigg(\frac{1}{N^{\ell-1}}\sum_{|j|_\ell=N-|M|_\ell}
		I^{j_1}\otimes
		\wick{
			\c1 a_{M_1}\otimes
			\ldots\otimes I^{j_\ell}\otimes
			\c1 a_{M_\ell}
		}\bigg)
		+R_N''\,,
	\end{multline*}
	where in the last line we used that for all $\varsigma\in\mathfrak{S}_\ell$ there are $\ell$ permutations which are equivalent to $\varsigma$ up to a cyclic permutation.
	Indeed, for any permutation of $a_{M_1},\ldots, a_{M_\ell}$ we may use the $\gamma_N$-invariance to write the corresponding contribution fixing the position of the factor $a_{M_\ell}$.
	This boils down to a permutation of $a_{M_1},\ldots, a_{M_{\ell-1}}$ which is repeated $\ell$ times.
	
	By induction this proves Equation \eqref{Eq: product of gamma sequences - estimate of remainder} for all $\ell\in\mathbb{N}$ and thus Equation \eqref{Eq: equivalence class of product of gamma sequences}.
\end{proof}

\begin{remark}\label{Rmk: estimate of commutator with the remainder; Cstar algebra of symmetric sequences; character space for algebra of gamma sequences}
	\noindent
	\begin{enumerate}[(i)]
        \item
        The appearance of the total weighted symmetrization \eqref{Eq: total weighted symmetrization} ensures that, when $a_{M_j}=I^{M_j}$ for all $j\in\{1,\ldots,\ell\}$, the right-hand side of Equation \eqref{Eq: equivalence class of product of gamma sequences} coincides with $[I^N]_N$.
        This is related to the fact that $\frac{1}{N^{\ell-1}}\sum_{|j|_\ell=N-|M|_\ell}=(\ell-1)!+O(1/N)$.
        
		\item\label{Item: estimate of commutator with the remainder}
		A closer inspection to the remainder term $R_N$ of Equation \eqref{Eq: product of gamma sequences - estimate of remainder} reveals that
		\begin{align}\label{Eq: estimate on commutator with the remainder}
			\big\|[R_N,
			\overline{\gamma}^{M_1'}_N(a_{M_1'})\ldots\overline{\gamma}^{M_{\ell'}'}_N(a_{M_{\ell'}'})]\big\|_N
			=O(1/N^2)\,,
		\end{align}
		for all $\ell', M_1',\ldots, M_{\ell'}'\in\mathbb{N}$, and $a_M\in B^M$.
		Roughly speaking, the reason for this is due to the estimate $\|R_N\|_N=O(1/N)$ together with the fact that both $(R_N)_N$ and \\
		$(\overline{\gamma}^{M_1'}_Na_{M_1'}\ldots\overline{\gamma}^{M_{\ell'}'}a_{M_{\ell'}'})_N$ are sequences with "an increasing number of identities".
		In more details, Equation \eqref{Eq: product of gamma sequences - estimate of remainder} implies
		\begin{align*}
			\overline{\gamma}^{M_1'}_N(a_{M_1'})
			\ldots\overline{\gamma}^{M_{\ell'}'}(a_{M_{\ell'}'})
			=\frac{1}{N^{\ell'-1}}\sum_{|j|_{\ell'}=N-|M'|_{\ell'}}
			\overline{\gamma}_N(I^{j_1}\otimes
			\wick{\c1 a_{M_1'}\otimes\ldots\otimes I^{j_{\ell'}}\otimes \c1 a_{M_{\ell'}'}})
			+R_N'\,,
		\end{align*}
		where $\|R_N'\|_N=O(1/N)$.
		This implies
		\begin{multline*}
			[R_N,
			\overline{\gamma}^{M_1'}_N(a_{M_1'})\ldots\overline{\gamma}^{M_{\ell'}'}_N(a_{M_{\ell'}'})]
			=[R_N,R_N']
			\\+\bigg[R_N,\frac{1}{N^{\ell'-1}}\sum_{|j|_{\ell'}=N-|M'|_{\ell'}}
			\overline{\gamma}_N(I^{j_1}\otimes
			\wick{\c1 a_{M_1'}\otimes\ldots\otimes I^{j_{\ell'}}\otimes \c1 a_{M_{\ell'}'}})\bigg]\,.
		\end{multline*}
		The first contribution is estimated by $\|[R_N,R_N']\|_N=O(1/N^2)$ while for the second contribution we have
		\begin{multline*}
			\bigg\|\bigg[R_N,\frac{1}{N^{\ell'-1}}\sum_{|j|_{\ell'}=N-|M'|_{\ell'}}
			\overline{\gamma}_N(I^{j_1}\otimes
			\wick{\c1 a_{M_1'}\otimes\ldots\otimes I^{j_{\ell'}}\otimes \c1 a_{M_{\ell'}'}})\bigg]\bigg\|
			\\\leq\frac{1}{N}\sum_{p=0}^{N-1}\bigg\|\bigg[R_N,\frac{1}{N^{\ell'-1}}\sum_{|j|_{\ell'}=N-|M'|_{\ell'}}
			\gamma_N^p(I^{j_1}\otimes
			\wick{\c1 a_{M_1'}\otimes\ldots\otimes I^{j_{\ell'}}\otimes \c1 a_{M_{\ell'}'}})\bigg]\bigg\|
			\\\leq\frac{L}{N^2}C_{M_1,\ldots,M_{\ell'}}\,.
		\end{multline*}
		where $C_{M_1,\ldots,M_{\ell'}}>0$ does not depend on $N$.
		In the last inequality we used the estimate $\|R_N\|_N=O(1/N)$ and that, on account of the structure of $R_N$ ---\textit{cf.} the proof of Proposition \ref{Prop: the C* algebra of gamma-sequences is commutative}--- and of $I^{j_1}\otimes a_{M_1'}\otimes\ldots\otimes I^{j_{\ell'}}\otimes a_{M_{\ell'}'}$, the sum over $p$ is non-vanishing for finitely many values, say $L$, where $L$ is $N$-independent.
		This proves Equation \eqref{Eq: estimate on commutator with the remainder}.
		\item 
		In complete analogy with Definition \ref{Def: C*-algebra of gamma-sequences} one may introduce the $C^*$-algebra $[B]^\infty_\pi\subset [B]_\sim$ generated by equivalence classes of $\pi$-sequences \cite[Def. II.1]{Raggio_Werner_1989}.
		Moreover, as shown in \cite{Landsman_2017,Landsman_Moretti_vandeVen_2020,Raggio_Werner_1989}, for any $a_M\in B^M_\pi$ and $a_{M'}\in B^{M'}_\pi$ one finds
		\begin{align}\label{Eq: pi-sequences are asymptotically closed under product}
			[\pi^M_N(a_M)\pi^{M'}_N(a_{M'})]_N
			=[\pi^{M+M'}_N(S_{M+M'}(a_M\otimes a_{M'}))]_N\,,
		\end{align}
		which shows that also $[B]^\infty_\pi$ is a commutative $C^*$-algebra.
		In fact, the product of $\pi$-sequences is (asymptotically as $N\to\infty$) a $\pi$-sequence.
		Additionally, one may prove that the system $\{B^N_\pi\}_{N\in\mathbb{N}},\{\pi^M_N\}_{N\geq M}$ is a \textbf{generalized inductive system} \cite{Blackadar_2006, Blackadar_Kirchberg_1997}.
		This streamlines the identification of a bundle of $C^*$-algebras $\prod_{N\in\overline{\mathbb{N}}}B^N_\pi$ out of which a strict deformation quantization can be constructed \cite{Landsman_Moretti_vandeVen_2020,Raggio_Werner_1989}.
		
		The situation for $\gamma$-sequences is slightly different.
		Indeed, Equation \eqref{Eq: equivalence class of product of gamma sequences} shows that the product of $\gamma$-sequences is not a $\gamma$-sequence, even if its $\sim$-equivalence class is considered. Nevertheless, Equation \eqref{Eq: equivalence class of product of gamma sequences} shows that the product of equivalence classes of $\gamma$-sequences is commutative.
		As we shall see in Section \ref{Sec: The continuous bundle of Cstar-algebras associated with gamma-sequences} this will be enough to identify a continuous bundle of $C^*$-algebras $\prod_{N\in\overline{\mathbb{N}}}B^N_\gamma$ out of which a strict deformation quantization is obtained.
		
		Finally it is worth observing that, for all $a\in B$, one finds $\overline{\gamma}^1_Na=\pi^1_Na$ so that, given the results of \cite{Landsman_Moretti_vandeVen_2020,Raggio_Werner_1989}, $[B]^\infty_\pi\subseteq [B]^\infty_\gamma$.
		
		\item
		By standard Gelfand duality \cite[$\S$ II.2]{Blackadar_2006} we find
		\begin{align*}
			[B]^\infty_\gamma\simeq C(K([B]^\infty_\gamma))\,,
		\end{align*}
		where $K([B]^\infty_\gamma)$ denotes the character space over $[B]^\infty_\gamma$.
		An element $\varphi\in K([B]^\infty_\gamma)$ is completely characterized by
		\begin{align*}
			\varphi_M(a_M):=\varphi([\overline{\gamma}^M_Na_M]_N)\,,
			\qquad
			a_M\in B^M\,,
		\end{align*}
		which identifies a sequence $\{\varphi_M\}_{M\in\mathbb{N}}$ of states with $\varphi_M\in S(B^M)$.
		These states are "asymptotically equivalent" because of Equation \eqref{Eq: gamma sequences asymptotic consistency}.
		Indeed considering $\overline{\gamma}_N^M\colon B^M\to B^N_\gamma$, $\overline{\gamma}_N^Ma_M=\overline{\gamma}_N(I^{N-M}\otimes a_M)$, we find
		\begin{align*}
			\lim_{N\to\infty}(\varphi_N\circ \overline{\gamma}_N^M)(a_M)
			=\lim_{N\to\infty}\varphi\left([\overline{\gamma}_{N'}^N\overline{\gamma}_N^Ma_M]_{N'}\right)
			=\varphi\left([\overline{\gamma}_{N'}^Ma_M]_{N'}\right)
			=\varphi_M(a_M)\,.
		\end{align*}	
		A similar argument goes for $[B]^\infty_\pi$, where $K([B]^\infty_\pi)$ can be explicitly characterized.
		In particular $K([B]^\infty_\pi)\simeq S(B)$ \cite[Lem. IV.4]{Raggio_Werner_1989}. This identification may also be seen as a consequence of the prominent quantum De Finetti Theorem \cite[Thm. 8.9]{Landsman_2017}.
		As shown in \cite{Landsman_Moretti_vandeVen_2020}, $S(B)$ is a stratified manifold which carries a Poisson structure.
	\end{enumerate}
\end{remark}

\section{Strict deformation quantization of $\gamma$-sequences}
\label{Sec: strict deformation quantization of gamma-sequences}

The goal of this section is to construct a strict deformation quantization of the commutative algebra $[B]^\infty_\gamma$.
%character space $X_\gamma:=K([B]^\infty_\gamma)$.
To this avail in Section \ref{Sec: The continuous bundle of Cstar-algebras associated with gamma-sequences} we will identity a suitable continuous bundle of $C^*$-algebras $[B]_\gamma$ by means of a standard construction \cite{Landsman_1998,Landsman_2017}.
In Section \ref{Subsec: canonical representative of dense elements in the algebra of gamma-sequence} we will introduce the notion of "canonical representative" for an element $[a_N]_N\in[\dot{B}]^\infty_\gamma$ ---\textit{cf.} Definitions \ref{Def: canonical representative of gamma-sequence}-\ref{Def: canonical representative of product of gamma-sequences}.
Eventually, in Section \ref{Subsec: the Poisson structure of the algebra of gamma-sequences} will show that $[B]^\infty_\gamma$ carries a Poisson structure and we will prove Theorem \ref{Thm: deformation quantization of the algebra of gamma sequences}, which provides the strict deformation quantization of $[B]^\infty_\gamma$.

\subsection{The continuous bundle of $C^*$-algebras $[B]_\gamma$ associated with $[B]^\infty_\gamma$}
\label{Sec: The continuous bundle of Cstar-algebras associated with gamma-sequences}

In this section we will define a continuous bundle of $C^*$-algebras $[B]_\gamma$ over $\overline{\mathbb{N}}=\mathbb{N}\cup\{\infty\}$ whose fibers are $[B]^N_\gamma:=B^N_\gamma$ for $N\in\mathbb{N}$ and $[B]^\infty_\gamma$, defined as per Definition \ref{Def: C*-algebra of gamma-sequences}, for $N=\infty$.

To this avail we briefly recall the main definitions and results we need ---\textit{cf.} \cite[App. C.19]{Landsman_2017}, \cite[\S IV.1.6]{Blackadar_2006}.
We denote by $C(\overline{\mathbb{N}})$ the space of $\mathbb{C}$-valued sequences $(\alpha_N)_{N\in\mathbb{N}}$ such that $\alpha_\infty:=\lim\limits_{N\to\infty}\alpha_N\in\mathbb{C}$ exists.
A \textbf{continuous bundle (or field) of $C^*$-algebras} over $\overline{\mathbb{N}}$ is a triple $\mathcal{A}$, $\{\mathcal{A}_N\}_{N\in\overline{\mathbb{N}}}$, $\{\psi_N\}_{N\in\overline{\mathbb{N}}}$ made by $C^*$-algebras $\mathcal{A},\mathcal{A}_N$, $N\in\overline{\mathbb{N}}$, and surjective homomorphisms $\psi_N\colon\mathcal{A}\to\mathcal{A}_N$ such that:
\begin{enumerate}[(i)]
	\item
	\label{Item: norm condition for bundle of Cstar algebras}
	The norm of $\mathcal{A}$ is given by $\|a\|_{\mathcal{A}}:=\sup\limits_{N\in\overline{\mathbb{N}}}\|\psi_N(a)\|_{\mathcal{A}_N}$;
	
	\item
	\label{Item: product by function condition for bundle of Cstar algebras}
	For all $\alpha=(\alpha_N)_{N\in\overline{\mathbb{N}}}\in C(\overline{\mathbb{N}})$ and $a\in\mathcal{A}$ there exists a $\alpha a\in\mathcal{A}$ with the property that $\psi_N(\alpha a)=\alpha_N\psi_N(a)$.
	
	\item
	\label{Item: norm continuity for bundle of Cstar algebras}
	For all $a\in\mathcal{A}$, $(\|\psi_N(a)\|_{\mathcal{A}_N})_{N\in\overline{\mathbb{N}}}\in C(\overline{\mathbb{N}})$.
\end{enumerate}
A \textbf{continuous section} of $\mathcal{A}$ is an element $a\in\prod_{N\in\overline{\mathbb{N}}}\mathcal{A}_N$ such that there exists $a'\in\mathcal{A}$ fulfilling $a_N=\psi_N(a')$ for all $N\in\overline{\mathbb{N}}$.
Clearly $\mathcal{A}$ can be identified with its continuous sections, therefore, in the forthcoming discussion we shall always regard $a\in\mathcal{A}$ as an element $a\in\prod_{N\in\overline{\mathbb{N}}}\mathcal{A}_N$.
For this reason from now on we will implicitly identify $\psi_N$, $N\in\overline{\mathbb{N}}$, with the projection $\prod_{N\in\overline{\mathbb{N}}}\mathcal{A}_N\to\mathcal{A}_N$.

\begin{remark}\label{Rmk: sufficient condition for continuous bundle of Cstar algebras}
    In applications, it is often difficult to identity a continuous bundle of $C^*$-algebras by assigning the triple $\mathcal{A}$, $\{\mathcal{A}_N\}_{N\in\overline{\mathbb{N}}}$, $\{\psi_N\}_{N\in\overline{\mathbb{N}}}$ directly.
    However, a useful result ---\textit{cf.} \cite[Prop. 1.2.3]{Landsman_1998}, \cite[Prop. C.124]{Landsman_2017}--- shows that it is in fact sufficient to identify a dense set of (a posteriori) continuous sections of $\mathcal{A}$.
    Actually, let $\widetilde{\mathcal{A}}\subseteq\prod_{N\in\overline{\mathbb{N}}}\mathcal{A}_N$ be such that:
    \begin{enumerate}
	    \item\label{Item: tildeA is pointwise dense}
	    For all $N\in\overline{\mathbb{N}}$ the set $\{a_N\,|\,a\in\widetilde{\mathcal{A}}\}$ is dense in $\mathcal{A}_N$;
	
	    \item\label{Item: tildeA is a star-algebra}
	    $\widetilde{\mathcal{A}}$ is a $*$-algebra;
	
	    \item\label{Item: tildeA fulfils the Rieffel condition}
	    For all $\widetilde{a}\in\widetilde{\mathcal{A}}$, it holds $\lim\limits_{N\to\infty}\|\widetilde{a}_N\|_{\mathcal{A}_N}=\|\widetilde{a}_\infty\|_{\mathcal{A}_\infty}$, \textit{i.e.} $(\|\widetilde{a}_N\|_{\mathcal{A}_N})_{N\in\overline{\mathbb{N}}}\in C(\overline{\mathbb{N}})$.
    \end{enumerate}
    Then defining $\mathcal{A}$ by
    \begin{align}\label{Eq: bundle defined from tildeA}
	    \mathcal{A}:=\bigg\lbrace a\in\prod_{N\in\overline{\mathbb{N}}}\mathcal{A}_N\,|\,
	    \forall\varepsilon>0\;\exists N_{\varepsilon}\in\mathbb{N}\,,\,
	    \exists a'\in\widetilde{\mathcal{A}}\colon
    	\|a_N-a'_N\|_{\mathcal{A}_N}<\varepsilon\;\forall N\geq N_\varepsilon
	    \bigg\rbrace\,,
    \end{align}
    one may prove that $\mathcal{A}$ is a continuous bundle of $C^*$-algebras over $\overline{\mathbb{N}}$ \cite{Landsman_1998,Landsman_2017}.
    In fact, $\mathcal{A}$ is the smallest continuous bundle of $C^*$-algebras over $\overline{\mathbb{N}}$ which contains $\widetilde{\mathcal{A}}$.
\end{remark}

We will now prove that $\prod_{N\in\overline{\mathbb{N}}}[B]^N_\gamma$ identifies a continuous bundle of $C^*$-algebras where $[B]^N_\gamma:=B^N_\gamma$ for $N\in\mathbb{N}$ while $[B]^\infty_\gamma$ denotes the $C^*$-algebra introduced in Definition \ref{Def: C*-algebra of gamma-sequences}.
To this avail we will identity a subset $\widetilde{\mathcal{A}}\subset\prod_{N\in\overline{\mathbb{N}}}[B]^N_\gamma$ fulfilling conditions \ref{Item: tildeA is pointwise dense}-\ref{Item: tildeA is a star-algebra}-\ref{Item: tildeA fulfils the Rieffel condition} of Remark \ref{Rmk: sufficient condition for continuous bundle of Cstar algebras}.
From a technical point of view, condition \ref{Item: tildeA fulfils the Rieffel condition} will require to prove that, for all $[a_N]_N\in[B]^\infty_\gamma$, the sequence $(\|a_N\|_N)_N$ has a limit as $N\to\infty$: This is proved in Proposition \ref{Prop: elements in the algebra of gamma-sequences norm have limit}.
To this avail, the following Lemma comes in handy.

\begin{lemma}\label{Lem: inequality for convergence of a sequence}
	Let $(\alpha_N)_{N\in\mathbb{N}}$ be a bounded sequence of real numbers such that
	\begin{align}\label{Eq: inequality for convergence of a sequence}
		\exists C_1,C_2\in\mathbb{R}\,,\,
		\exists N_0\in\mathbb{N}\;\colon\;
		\alpha_N\geq\alpha_K+C_1\frac{1}{K}+C_2\frac{K}{N}\;
		\forall N\geq K\geq N_0\,.
	\end{align}
	Then $(\alpha_N)_N\in C(\overline{\mathbb{N}})$, \textit{i.e.} $\alpha_\infty:=\lim\limits_{N\to\infty}\alpha_N\in\mathbb{R}$ exists.
\end{lemma}
\begin{proof}
    Let $(\alpha_{N_j})_{j\in\mathbb{N}}$ be a convergent subsequence of $(\alpha_N)_{N\in\mathbb{N}}$.
    Then for all $K\geq N_0$ we find
    \begin{align*}
        \lim_{j\to\infty}\alpha_{N_j}
        \geq\lim_{j\to\infty}\bigg(\alpha_K+C_1\frac{1}{K}+C_2\frac{K}{N_j}\bigg)
        =\alpha_K+\frac{C_1}{K}\,.
    \end{align*}
    Since this holds true for all convergent subsequences we conclude that
    \begin{align*}
        \liminf_{N\to\infty}\alpha_N
        \geq\alpha_K+\frac{C_1}{K}
        \qquad\forall K\geq N_0\,.
    \end{align*}
    Considering again a convergent subsequence $(\alpha_{K_j})_{j\in\mathbb{N}}$ of $(\alpha_N)_{N\in\mathbb{N}}$ the above inequality implies
    \begin{align*}
        \lim_{j\to\infty}\alpha_{K_j}
        \leq\lim_{j\to\infty}\bigg(\liminf_{N\to\infty}\alpha_N-\frac{C_1}{K_j}\bigg)
        =\liminf_{N\to\infty}\alpha_N\,.
    \end{align*}
    Since this holds for all convergent subsequences we conclude that
    \begin{align*}
        \limsup_{N\to\infty}\alpha_N
        \leq\liminf_{N\to\infty}\alpha_N\,,
    \end{align*}
    therefore, $\lim\limits_{N\to\infty}\alpha_N=\liminf\limits_{N\to\infty}\alpha_N=\limsup\limits_{N\to\infty}\alpha_N$ exists and it is finite.
\end{proof}

\noindent

\begin{proposition}\label{Prop: elements in the algebra of gamma-sequences norm have limit}
	For all $[a_N]_N\in [B]^\infty_\gamma$ the sequence $(\|a_N\|_N)_N$ is in $C(\overline{\mathbb{N}})$.
	In particular we have
	\begin{align*}
		\|[a_N]_N\|_{[B]^\infty_\gamma}
		(=\|[a_N]_N\|_{[B]_\sim})
		=\lim\limits_{N\to\infty}\|a_N\|_N\,.
	\end{align*}
\end{proposition}
\begin{proof}
	To begin with we prove the claim for $[a_N]_N=[\overline{\gamma}^M_Na_M]_N$.
	We will then move to $[a_N]_N\in[\dot{B}]^\infty_\gamma$ eventually discussing $[a_N]_N\in [B]^\infty_\gamma$.
	
	\begin{description}
		\item[$\boxed{\overline{\gamma}^M_Na_M}$]
		Let $a_M\in B^M$, $M\in\mathbb{N}$, and let us consider $[\overline{\gamma}^M_Na_M]_N$.
		Let $N,K\in\mathbb{N}$, $N\geq K\geq M$ and consider $\omega_K\in S(B^K)$.
		We decompose $\omega_K$ in a finite convex combination of product states
		\begin{align*}
			\omega_K=\sum_{p_1,\ldots,p_K}\omega_K^{p_1\ldots p_K}
			\eta_{p_1}\otimes\ldots\otimes\eta_{p_K}\,,
		\end{align*}
		where $\eta_{p_\ell}\in S(B)$ for all $\ell=1,\ldots,K$.
		We then consider
		\begin{align*}
			\omega_{K,N}:=\sum_{p_1,\ldots,p_K}\omega_K^{p_1\ldots p_K}
			\tau^r\otimes
			(\eta_{p_1}\otimes\ldots\otimes\eta_{p_K})^q\in S(B^N)\,,
		\end{align*}
		where $N=r+qK$, $q\in\mathbb{N}$ and $r\in\{0,\ldots,K-1\}$ while $\tau\in S(B)$ is normalized the trace state.
		We consider
		\begin{align*}
			\omega_{K,N}(\overline{\gamma}^M_Na_M)
			=\omega_{K,N}
			\bigg(\frac{1}{N}\sum_{j=0}^{N-1}I^{N-M-j}\otimes a_M\otimes I^j\bigg)\,.
		\end{align*}
		By direct inspection we have that, for all $\ell\in\{0,\ldots,K-1\}$,
		\begin{multline*}
			\frac{1}{N}\bigg[\sum_{p_1,\ldots,p_K}\omega_K^{p_1\ldots p_K}
			\tau^r\otimes
			(\eta_{p_1}\otimes\ldots\otimes\eta_{p_K})^q\bigg]
			\bigg(I^{N-M-\ell}\otimes a_M\otimes I^\ell\bigg)
			\\=\frac{1}{N}\omega_K(\gamma_K^\ell (I^{K-M}\otimes a_M))\,.
		\end{multline*}
		The same contribution arises if $j\leq N-r-M=qK-M$ and $j=\ell\mod K$.
		The number of such $j$'s is roughly
		\begin{align*}
			q-M/K=N/K-r/K-M/K=N/K+O(1)\,,
		\end{align*}
		where the $O(1)$ contribution is bounded both in $N$ and in $K$.
		The net result is
		\begin{multline*}
			\omega_{K,N}(\overline{\gamma}^M_Na_M)
			=\sum_{\ell=0}^{K-1}\bigg(\frac{N}{K}+O(1)\bigg)\frac{1}{N}\omega_K(\gamma_K^\ell (I^{K-M}\otimes a_M))
			\\+\frac{1}{N}\sum_{j=N-r-M+1}^{N-1}
			\omega_{K,N}\big(\overline{\gamma}_N^j(I^{N-M}\otimes a_M)\big)
			\\=\omega_K(\overline{\gamma}_K^Ma_M)
			+O(K/N)\,,
		\end{multline*}
		where we observed that the sum over $j\in[N-r-M+1,N-1]$ contains at most $r+M-1=O(K)$ terms each of which is bounded by $\|a_M\|_M$.
		Overall we find
		\begin{align*}
			\|\overline{\gamma}^M_Na_M\|_N
			\geq|\omega_{K,N}(\overline{\gamma}^M_Na_M)|
			&=\left|\omega_K(\overline{\gamma}^M_Ka_M)+C\frac{K}{N}\right|
			\geq|\omega_K(\overline{\gamma}^M_Ka_M)|
			-C\frac{K}{N}\,,
		\end{align*}
		where $C>0$ depends on $M$ but not on $N$ or $K$.
		The arbitrariness of $\omega_K\in S(B^K)$ leads to
		\begin{align*}
			\|\overline{\gamma}^M_Na_M\|_N
			\geq\|\overline{\gamma}^M_Ka_M\|_K
			-C\frac{K}{N}\,.
		\end{align*}
		Thus, Lemma \ref{Lem: inequality for convergence of a sequence} applies to the sequence $(\|\overline{\gamma}^M_Na_M\|_N)_N$ proving that $\lim\limits_{N\to\infty}\|\overline{\gamma}^M_Na_M\|_N$ exists.
		
		\item[$\boxed{[\dot{B}]^\infty_\gamma}$]
		We now consider an arbitrary element $[a_N]_N$.
		Although our proof works for an arbitrary element of $[a_N]_N\in[\dot{B}]^\infty_\gamma$, for the sake of (notational) simplicity we restrict ourself to the case
		\begin{align*}
			[a_N]_N
			&=\bigg[\sum_{k_1,k_2}c^{k_1k_2}
			\overline{\gamma}^{M(k_1)}_N(a_{(k_1)})
			\overline{\gamma}^{M(k_2)}_N(a_{(k_2)})\bigg]_N\,,
		\end{align*}
		where the sum over $k_1,k_2$ is finite.
		To prove that $(\|a_N\|_N)_N$ has a limit as $N\to\infty$ we rely on Equation \eqref{Eq: equivalence class of product of gamma sequences} together with an argument similar in spirit to the one used for the case of a single $\gamma$-sequence.
		In fact, Proposition \ref{Prop: the C* algebra of gamma-sequences is commutative} implies that
		\begin{multline*}
			\bigg\|\sum_{k_1,k_2}c^{k_1k_2}
			\overline{\gamma}^{M(k_1)}_N(a_{(k_1)})
			\overline{\gamma}^{M(k_2)}_N(a_{(k_2)})\bigg\|_N
			\\=\bigg\|\sum_{k_1,k_2}c^{k_1k_2}
			\frac{1}{N}\sum_{j_1+j_2=N-M(k_1)-M(k_2)}
			\overline{\gamma}_N\left(
			I^{j_1}\otimes a_{(k_1)}\otimes
			I^{j_2}\otimes a_{(k_2)}
			\right)\bigg\|_N
			+O(1/N)\,,
		\end{multline*}
		so that we may restrict to the first factor on the right-hand side.
		
		As for the case of single $\gamma$-sequence let $N,K\in\mathbb{N}$ be such that $N\geq K\geq \max\limits_{k_1,k_2}\{2(M(k_1)+M(k_2))\}$ where the maximum is taken over all pairs $k_1,k_2\in\mathbb{N}$ appearing in the sum defining $[a_N]_N$.
		We consider $\omega_K\in S(B^K)$ and, as above, we set
		\begin{align*}
			\omega_{N,K}:=\sum_{p_1,\ldots,p_K}\omega_K^{p_1\ldots p_K}
			\tau^r\otimes(\eta_{p_1}\otimes\ldots\otimes\eta_{p_K})^q
			\in S(B^N)\,,
		\end{align*}
		where $N=r+qK$, $q\in\mathbb{N}$ and $r\in\{0,\ldots,K-1\}$ while $\omega_K=\sum\limits_{p_1,\ldots,p_K}\omega_K^{p_1\ldots p_K}\eta_{p_1}\otimes\ldots\otimes\eta_{p_K}$ is an arbitrary finite convex decomposition of $\omega_K$ into product states.
		We then evaluate
		\begin{multline*}
			\omega_{K,N}\bigg(
			\sum_{k_1,k_2}c^{k_1k_2}
			\frac{1}{N}\sum_{j_1+j_2=N-M(k_1)-M(k_2)}
			\overline{\gamma}_N\big(
			I^{j_1}\otimes a_{(k_1)}\otimes
			I^{j_2}\otimes a_{(k_2)}
			\big)
			\bigg)
			\\=\sum_{k_1,k_2}c^{k_1k_2}
			\frac{1}{N}\sum_{j_1+j_2=N-M(k_1)-M(k_2)}
			\omega_{K,N}\left(\overline{\gamma}_N\big(
			I^{j_1}\otimes a_{(k_1)}\otimes
			I^{j_2}\otimes a_{(k_2)}
			\big)\right)\,.
		\end{multline*}
		To this avail we fix $k_1,k_2$ and split the sum over $j_2$ in two cases:
		\begin{enumerate}[(a)]
			\item
			Let consider the sum for $0\leq j_2\leq N-M(k_1)-M(k_2)-r$.
			For $0\leq \ell\leq K-M(k_1)-M(k_2)$ we find, with the same argument used for a single $\gamma$-sequence,
			\begin{multline*}
				\omega_{K,N}\bigg(\overline{\gamma}_N\big(
				I^{N-M(k_1)-M(k_2)-\ell}\otimes a_{(k_1)}\otimes
				I^{\ell}\otimes a_{(k_2)}\big)\bigg)
				\\=\omega_K\bigg(
				\overline{\gamma}_K(I^{K-M(k_1)-M(k_2)-\ell}\otimes a_{(k_1)}
				\otimes I^{\ell}\otimes a_{(k_2)})\bigg)
				+O(K/N)\,,
			\end{multline*}
			The number of $j_2$'s such that $0\leq j_2\leq N-M(k_1)-M(k_2)-r$ and $j_2=\ell\mod K$ is roughly $q=N/K+O(1)$, therefore, summing over such $j_2$'s leads to a contribution of
			\begin{multline*}
				\omega_{K,N}\bigg(
				\frac{1}{N}
				\sum_{\substack{j_1+j_2=N-M(k_1)-M(k_2) \\ 0\leq j_2\leq N-r-M(k_1)-M(k_2)\\ j_2\leq K-M(k_1)-M(k_2) \mod K}}
				\overline{\gamma}_N\big(
				I^{j_1}\otimes a_{(k_1)}\otimes
				I^{j_2}\otimes a_{(k_2)}
				\big)
				\bigg)
				\\=\omega_K\bigg(
				\frac{1}{K}\sum_{j_1+j_2=K-M(k_1)-M(k_2)}
				\overline{\gamma}_K(I^{j_1}\otimes a_{(k_1)}
				\otimes I^{j_2}\otimes a_{(k_2)})
				\bigg)
				+O(K/N)\,.
			\end{multline*}
			It remains to discuss the sum over $0\leq j_2\leq N-M(k_1)-M(k_2)-r$ with $j_2\in[K-M(k_1)-M(k_2),K-1]\mod K$: In this case we find
			\begin{multline*}
				\bigg|\omega_{K,N}\bigg(
				\frac{1}{N}\sum_{\substack{j_1+j_2=N-M(k_1)-M(k_2)\\0\leq j_2\leq N-M(k_1)-M(k_2)-r\\K-M(k_1)-M(k_2)\leq j_2\leq K-1\mod K}}
				\overline{\gamma}_N\big(I^{j_1}\otimes a_{(k_1)}\otimes I^{j_2}\otimes a_{(k_2)}\big)
				\bigg)\bigg|
				\\\leq\frac{1}{N}
				\sum_{\substack{j_1+j_2=N-M(k_1)-M(k_2)\\0\leq j_2\leq N-M(k_1)-M(k_2)-r\\K-M(k_1)-M(k_2)\leq j_2\leq K-1\mod K}}
				\|a_{(k_1)}\|_{M(k_1)}\|a_{(k_2)}\|_{M(k_2)}
				\\\leq\frac{1}{N}\bigg(\frac{N}{K}+O(1)\bigg)
				(M(k_1)+M(k_2)-1)
				\|a_{(k_1)}\|_{M(k_1)}\|a_{(k_2)}\|_{M(k_2)}
				=O(1/K)\,,
			\end{multline*}
			where we observed that, for each of the $M(k_1)+M(k_2)-1$ values of $\ell\in[K-M(k_1)-M(k_2),K-1]$, there are $q=N/K+O(1)$ values of $j_2\leq N-M(k_1)-M(K_2)-r$ such that $j_2=\ell\mod K$.
			Loosely speaking these contributions arise when $j_2$ is such that "$a_{(k_2)}$ overlaps with the (translated) position of $a_{(k_1)}$".
			This does not allow to reconstruct $\omega_K$, therefore, these cases are estimated by $O(1/K)$.
			
			\item
			If $j_2\in[N-M(k_1)-M(k_2)-r+1,N-M(k_1)-M(k_2)]$ ---which is empty if $r=0$--- we have
			\begin{multline*}
				\omega_{K,N}\bigg(
				\frac{1}{N}
				\sum_{\substack{j_1+j_2=N-M(k_1)-M(k_2)\\ N-M(k_1)-M(k_2)-r\leq j_2\leq N-M(k_1)-M(k_2)}}
				\overline{\gamma}_N\big(
				I^{j_1}\otimes a_{(k_1)}\otimes
				I^{j_2}\otimes a_{(k_2)}
				\big)\bigg)
				\\=O(K/N)\,.
			\end{multline*}
		\end{enumerate}
		
		Recollecting our result we have
		\begin{multline*}
			\bigg|\omega_{K,N}\bigg(
			\sum_{k_1,k_2}c^{k_1k_2}
			\frac{1}{N}\sum_{j_1+j_2=N-M(k_1)-M(k_2)}
			\overline{\gamma}_N\big(
			I^{j_1}\otimes a_{(k_1)}\otimes
			I^{j_2}\otimes a_{(k_2)}
			\big)\bigg)\bigg|
			\\=\bigg|\omega_K\bigg(
			\sum_{k_1,k_2}c^{k_1k_2}
			\frac{1}{K}\sum_{j_1+j_2=K-M(k_1)-M(k_2)}
			\overline{\gamma}_K(I^{j_1}\otimes a_{(k_1)}
			\otimes I^{j_2}\otimes a_{(k_2)})
			\bigg)
			+O(1/K)
			+O(K/N)\bigg|
			\\=\bigg|\omega_K\bigg(
			\sum_{k_1,k_2}c^{k_1k_2}
			\frac{1}{K}\sum_{j_1+j_2=K-M(k_1)-M(k_2)}
			\overline{\gamma}_K(I^{j_1}\otimes a_{(k_1)}
			\otimes I^{j_2}\otimes a_{(k_2)})
			\bigg)\bigg|
			-\frac{C_1}{K}
			-C_2\frac{K}{N}\,.
		\end{multline*}
		where $C_1,C_2>0$ do not depend neither on $N$ nor on $K$.
		The arbitrariness of $\omega_K\in S(B^K)$ leads to
		\begin{multline*}
			\bigg\|\sum_{k_1,k_2}\frac{1}{N}\sum_{j_1+j_2=N-M(k_1)-M(k_2)}
			\overline{\gamma}_N\bigg(
			I^{j_1}\otimes a_{(k_1)}\otimes
			I^{j_2}\otimes a_{(k_2)}
			\bigg)\bigg\|_N
			\\\geq\bigg\|\sum_{k_1,k_2}\frac{1}{K}\sum_{j_1+j_2=K-M(k_1)-M(k_2)}
			\overline{\gamma}_K\bigg(
			I^{j_1}\otimes a_{(k_1)}\otimes
			I^{j_2}\otimes a_{(k_2)}
			\bigg)\bigg\|_K-\frac{C_1}{K}-C_2\frac{K}{N}\,.
		\end{multline*}
		Thus, Lemma \ref{Lem: inequality for convergence of a sequence} applies and the limit
		\begin{align*}
			\lim_{N\to\infty}\|\overline{\gamma}_N^{M(k_1)}(a_{(k_1)})
			\overline{\gamma}_N^{M(k_2)}(a_{(k_2)})\|_N\,,
		\end{align*}
		exists and it is finite.
		
		\item[$\boxed{[B]^\infty_\gamma}$]
		Finally, let $[a_N]_N\in [B]^\infty_\gamma$.
		Then, for all $\varepsilon>0$ there exists $N_\varepsilon\in\mathbb{N}$ and $[a_N']_N\in[\dot{B}]^\infty_\gamma$ such that
		\begin{align*}
			\|a_N-a_N'\|_N<\varepsilon
			\qquad\forall N\geq N_\varepsilon\,.
		\end{align*}
		Moreover, since $(\|a_N'\|_N)_N$ is convergent, there exists $N_\varepsilon'\in\mathbb{N}$ such that
		\begin{align*}
			\left|\|a_N'\|_N-\|a_M'\|_M\right|<\varepsilon
			\qquad\forall N,M\geq N_\varepsilon'\,.
		\end{align*}
		For $N,M\geq \max\{N_\varepsilon,N_\varepsilon'\}$ we then have
		\begin{align*}
			\left|\|a_N\|_N-\|a_M\|_M\right|
			&\leq\left|\|a_N\|_N-\|a_N'\|_N\right|
			+\left|\|a_N'\|_N-\|a_M'\|_M\right|
			+\left|\|a_M'\|_M-\|a_M\|_M\right|
			\\&\leq\|a_N-a_N'\|_N
			+\left|\|a_N'\|_N-\|a_M'\|_M\right|
			+\|a_M'-a_M\|_M
			\leq 3\varepsilon\,,
		\end{align*}
		proving that $(\|a_N\|_N)_N$ is a Cauchy sequence.
	\end{description}
\end{proof}

\begin{remark}
	The result of Proposition \ref{Prop: elements in the algebra of gamma-sequences norm have limit} applies also for $\pi$-sequences.
	For this latter case the proof streamlines because
	\begin{align*}
		\|\pi^M_Na_M\|_N
		=\|\pi_N^K\pi^M_Ka_M\|_N
		\leq\|\pi^M_Ka_M\|_K\,,
	\end{align*}
	so that $(\|\pi^M_Na_M\|)_N$ is decreasing.
	The difficulties in moving from $[B]^\infty_\pi$ to $[B]^\infty_\gamma$ is twofold.
	On the one hand, for $\gamma$-sequences $\|\overline{\gamma}^M_Na_M\|_N$ is not decreasing, although it fulfils a similar properties asymptotically.
	On the other hand, the product of $\gamma$-sequences is not a $\gamma$-sequence, even when equivalence classes are considered.
	This requires a different strategy to ensure the existence of the limit $\lim\limits_{N\to\infty}\|a_N\|_N$ for $[a_N]_N\in [B]^\infty_\gamma$.
\end{remark}

The following proposition proves the existence of the continuous bundle of $C^*$-algebras of interest.
\begin{proposition}\label{Prop: continuous bundle of Cstar algebra for gamma-sequences}
	Let $\{B^N_\gamma\}_{N\in\mathbb{N}}$ be the family of $C^*$-algebras introduced in Equation \eqref{Eq: gamma invariant subalgebra}.
	Let $\{[B]^N_\gamma\}_{N\in\overline{\mathbb{N}}}$ be defined by $[B]^N_\gamma:=B^N_\gamma$ for $N\in\mathbb{N}$ while $[B]^\infty_\gamma$ is the $C^*$-algebra generated by equivalence classes of $\gamma$-sequences, \textit{cf.} Definition \ref{Def: C*-algebra of gamma-sequences}.
	Let $[\dot{B}]_\gamma\subset\prod_{N\in\overline{\mathbb{N}}}[B]^N_\gamma$ be the subset defined by
	\begin{align}\label{Eq: a posteriori continuous sections of the gamma-bundle}
		[\dot{B}]_\gamma:=\bigg\lbrace (A_N)_{N\in\overline{\mathbb{N}}}\in\prod_{N\in\overline{\mathbb{N}}}[B]^N_\gamma\,|\,
		\exists (a_N)_N\in \dot{B}^\infty_\gamma\colon
		A_N=\begin{cases}
			a_N& N\in\mathbb{N}
			\\
			[a_N]_N& N=\infty
		\end{cases}
		\bigg\rbrace\,.
	\end{align}
	Then $[\dot{B}]_\gamma$ fulfils conditions \ref{Item: tildeA is pointwise dense}-\ref{Item: tildeA is a star-algebra}-\ref{Item: tildeA fulfils the Rieffel condition} and thus it leads to a continuous bundle of $C^*$-algebras
	\begin{align}\label{Eq: gamma-bundle}
		[B]_\gamma:=\bigg\lbrace (A_N)_{N\in\overline{\mathbb{N}}}\in\prod_{N\in\overline{\mathbb{N}}}[B]^N_\gamma\,|\,
		\forall\varepsilon>0\;\exists N_{\varepsilon}\in\mathbb{N}\,,\,
		\exists A'\in[\dot{B}]_\gamma\colon
		\|A_N-A'_N\|_N<\varepsilon\;\forall N\geq N_\varepsilon
		\bigg\rbrace\,.
	\end{align}
\end{proposition}
\begin{proof}
    We will prove conditions \ref{Item: tildeA is pointwise dense}-\ref{Item: tildeA is a star-algebra}-\ref{Item: tildeA fulfils the Rieffel condition} of Remark \ref{Rmk: sufficient condition for continuous bundle of Cstar algebras}.
	The space $[\dot{B}]_\gamma$ is a $*$-algebra, therefore, condition \ref{Item: tildeA is a star-algebra} is fulfilled.
	Concerning condition \ref{Item: tildeA is pointwise dense}, we have to prove that
	\begin{align*}
		Z_M:=\{A_M\in[B]^M_\gamma\,|\,(A_N)_{N\in\overline{\mathbb{N}}}\in[\dot{B}]_\gamma\}
		\subseteq[B]^M_\gamma\,,
	\end{align*}
	is dense in $[B]^M_\gamma$ for all $M\in\overline{\mathbb{N}}$.
	For $M\in\mathbb{N}$ it is enough to observe that, for all $a_M\in [B]^M_\gamma=B^M_\gamma$, we may consider $(A_N)_{N\in\overline{\mathbb{N}}}\in [\dot{B}]_\gamma$ defined by
	\begin{align*}
		A_N=\begin{cases}
			\overline{\gamma}^M_Na_M& N\in\mathbb{N}
			\\
			[\overline{\gamma}_N^Ma_M]_N& N=\infty
		\end{cases}\,,
	\end{align*}
	which leads to $A_M=\overline{\gamma}_Ma_M=a_M$, \textit{i.e.} $Z_M=[B]^M_\gamma$.
	If $M=\infty$ we have $Z_\infty=[\dot{B}]^\infty_\gamma$ whose closure is per definition $[B]^\infty_\gamma$ ---\textit{cf.} Definition \ref{Def: C*-algebra of gamma-sequences}. 
	
	Finally condition \ref{Item: tildeA fulfils the Rieffel condition} is equivalent to
	\begin{align*}
		\lim_{N\to\infty}\|A_N\|_N
		=\lim_{N\to\infty}\|a_N\|_N
		=\|[a_N]_N\|_{[B]^\infty_\gamma}
		=\|A_\infty\|_{[B]^\infty_\gamma}
		\qquad\forall (A_N)_{N\in\overline{\mathbb{N}}}\in[\dot{B}]_\gamma\,,
	\end{align*}
	where the existence of $\lim\limits_{N\to\infty}\|a_N\|_N$ is ensured by Proposition \ref{Prop: elements in the algebra of gamma-sequences norm have limit}.
\end{proof}

\subsection{Canonical representative of $[a_N]_N\in[\dot{B}]^\infty_\gamma$}
\label{Subsec: canonical representative of dense elements in the algebra of gamma-sequence}

To proceed further in the construction of the deformation quantization of $[B]^\infty_\gamma$ we have to discuss the possibility of identifying a canonical representative of an element $[a_N]_N\in[\dot{B}]^\infty_\gamma$.
This is required for both endowing $[B]^\infty_\gamma$ with a Poisson structure as well as for defining the quantization maps $Q_N\colon[\dot{B}]^\infty_\gamma\to [B]^N_\gamma$ ---\textit{cf.} Theorem \ref{Thm: deformation quantization of the algebra of gamma sequences}.

To begin with, we address the following problem: Given $[\overline{\gamma}^M_Na_M]_N\in[\dot{B}]^\infty_\gamma$ does it hold
\begin{align*}
	[\overline{\gamma}_N^Ma_M]_N=[0]_N
	\Longleftrightarrow
	a_M=0\;?
\end{align*}
A positive answer in this direction would imply that, given an equivalence class $[\overline{\gamma}^M_Na_M]_N$, one is able to determine uniquely the $\gamma$-sequence $(\overline{\gamma}^M_Na_M)_N$.
Unfortunately, the answer to this question is negative because
\begin{align*}
	[\overline{\gamma}_N^M a_M]_N
	=[\overline{\gamma}_N^{M+K}(I^K\otimes a_M)]_N
	=[\overline{\gamma}_N^{M+K}(a_M\otimes I^K)]_N\,,
\end{align*}
although the associated sequences are not the same.
Indeed
\begin{align*}
	\overline{\gamma}_M^{M+K}(I^K\otimes a_M)=0
	\neq\overline{\gamma}_Ma_M=\overline{\gamma}^M_Ma_M\,.
\end{align*}
This counterexample suggests to focus on the $C^*$-subalgebra $\widetilde{B}^M$ where $\widetilde{B}=\ker\tau$, $\tau\in S(B)$ being the trace state ---\textit{cf.} Section \ref{Sec: The algebra of gamma-sequences}.
In fact, therein the situation is slightly better as shown by the following Lemma.

\begin{lemma}\label{Lem: Btilde gamma-sequences retain information on the initial value}
	Let $\widetilde{a}_M\in\widetilde{B}^M$ be such that $[\overline{\gamma}^M_N\widetilde{a}_M]_N=[0]_N$.
	Then $\widetilde{a}_M=0$.
\end{lemma}
\begin{proof}
	Per definition $[\overline{\gamma}^M_N\widetilde{a}_M]_N=[0]_N$ if and only if $\lim\limits_{N\to\infty}\|\overline{\gamma}^M_N\widetilde{a}_M\|_N=0$.
	Let $\omega_M\in S(B^M)$ and let $\tau\in S(B)$ be the normalized trace state $\tau(a):=\operatorname{tr}(a)/\kappa$.
%	Notice that, on account of \eqref{Eq: Mk-basis properties}, $\tau(\widetilde{a})=0$ for all $\widetilde{a}\in\widetilde{B}$.
	Let $N\geq M+1$, $q\in\mathbb{N}$ and $r\in\{0,\ldots,M\}$ be such that $N=r+q(M+1)$.
	We consider the state
	\begin{align*}
		\omega_{M,N}:=\tau^r\otimes(\tau\otimes\omega_M)^q\in S(B^N)\,.
	\end{align*}
	By direct inspection we find that
	\begin{align}\label{Eq: action of periodic state with trace on a gamma-sequence}
		\omega_{M,N}(\overline{\gamma}_N^M\widetilde{a}_M)
		=[\tau^r\otimes(\tau\otimes\omega_M)^q]\bigg(\frac{1}{N}\sum_{j=0}^{N-1}\gamma_N^j(I^{N-M}\otimes \widetilde{a}_M)\bigg)
		=\frac{1}{M+1}\omega_M(\widetilde{a}_M)
		+O(1/N)\,,
	\end{align}
	Indeed, for $j=0$ the resulting contribution is $\omega_M(\widetilde{a}_M)/N$.
	The same contribution appears when $j=0 \mod M+1$: Since $j\in\{0,\ldots,N-1\}$ this happens $q$ times, moreover, $q=N/(M+1)+O(1)$ leading to the right-hand side of Equation \eqref{Eq: action of periodic state with trace on a gamma-sequence}.
	Whenever $j\neq 0\mod M+1$ the resulting contribution is $0$, on account of the fact that $\tau$ vanishes on $\widetilde{B}$.
	
	Equation \eqref{Eq: action of periodic state with trace on a gamma-sequence} implies that, for all $\omega_M\in S(B^M)$,
	\begin{align*}
		0=\lim_{N\to\infty}\|\overline{\gamma}_N^M\widetilde{a}_M\|
		\geq\frac{1}{M+1}|\omega_M(\widetilde{a}_M)|\,.
	\end{align*}
	The arbitrariness of $\omega_M$ leads to $\|\widetilde{a}_M\|_M=0$, that is, $\widetilde{a}_M=0$.
\end{proof}

Thus, although the equivalence class $[\overline{\gamma}^M_Na_M]_N$ does not identify a unique sequence $(\overline{\gamma}^M_Na_M)_N$, Lemma \ref{Lem: Btilde gamma-sequences retain information on the initial value} suggests that a (a posteriori unique) canonical representative may be extracted by working with the "$\widetilde{B}$-irreducible components" of the $\gamma$-sequence.
To this avail, we introduce the notion of $\widetilde{B}$-irreducibility.
This identifies those elements in $B^M$ which cannot be written as $I\otimes a_{M-1}$ or $a_{M-1}\otimes I$ for some $a_{M-1}\in B^{M-1}$.

\begin{definition}\label{Def: Btilde-irreducible element}
	An element $a_M\in B^M$ is called $\widetilde{B}$-\textbf{irreducible}, and we write $a_M\in B^M_{\textsc{irr}}$, if either $M=0$ or
	\begin{align}\label{Eq: Btilde-irreducible element}
		(\tau\otimes\omega_{M-1})(a_M)
		=(\omega_{M-1}\otimes\tau)(a_M)=0\,,
	\end{align}
	for all $\omega_{M-1}\in S(B^{M-1})$.
\end{definition}
\begin{remark}
	Notice that, per definition, $a_0\in\mathbb{C}$ is $\widetilde{B}$-irreducible, moreover, $B^1_{\textsc{irr}}=\widetilde{B}$, $B^2_{\textsc{irr}}=\widetilde{B}^2$.
	For the sake of completeness, Appendix \ref{App: characterization of tildeB irredubible elements} provides a complete characterization of $B^M_{\textsc{irr}}$ for all $M\in\mathbb{N}$.
\end{remark}

The notion of $\widetilde{B}$-irreducible elements leads to a proper definition of "canonical representative" for a $\gamma$-sequence ---\textit{cf.} Definition \ref{Def: canonical representative of gamma-sequence}.
Indeed, let consider an arbitrary $a_M\in B^M$.
By considering a basis $I,b_1,\ldots, b_{\kappa^2-1}$ of $B$ fulfilling \eqref{Eq: Mk-basis properties} we may decompose $a_M$ as
\begin{multline}\label{Eq: decomposition of BM into part with a different number of identities}
	a_M=a_0I^M
	+\sum_{\substack{j_1+j_2=M-1 \\ k}}c^k_{j_1j_2} I^{j_1}\otimes b_k\otimes I^{j_2}
	+\sum_{\substack{j_1+j_2+j_3=M-2 \\ k_1,k_2}}c^{k_1k_2}_{j_1j_2j_3}
	I^{j_1}\otimes b_{k_1}
	\otimes I^{j_2}\otimes b_{k_2}\otimes I^{j_3}
	\\+\ldots+\sum_{\substack{j_1+\ldots+j_{\ell+1}=M-\ell \\ k_1,\ldots,k_\ell}}c^{k_1\ldots k_\ell}_{j_1\ldots j_{\ell+1}}
	I^{j_1}\otimes b_{k_1}
	\otimes\ldots
	\otimes I^{j_\ell}\otimes b_{k_\ell}\otimes I^{j_{\ell+1}}
	\\+\ldots
	+\sum_{k_1,\ldots,k_M}c^{k_1\ldots k_M}b_{k_1}\otimes\ldots\otimes b_{k_M}\,,
\end{multline}
where $a_0,c^{k_1\ldots k_\ell}_{j_1\ldots j_{\ell+1}}\in\mathbb{C}$ and the sum over $k_1,\ldots,k_\ell$ is finite.
At this stage we observe that $(\overline{\gamma}^M_Na_M)_N=(\overline{\gamma}^M_Na_M')_N$ where $a_M'\in B^M$ is defined by
\begin{align}
	\nonumber
	a_M'&=a_0I^M
	+I^{M-1}\otimes \sum_{\substack{j_1+j_2=M-1 \\ k}}c^k_{j_1j_2}b_k
	+\sum_{\substack{j_1+j_2+j_3=M-2 \\ k_1,k_2}}c^{k_1k_2}_{j_1j_2j_3}
	I^{j_1+j_3}\otimes b_{k_1}
	\otimes I^{j_2}\otimes b_{k_2}
	\\\nonumber
	&+\ldots+\sum_{\substack{j_1+\ldots+j_{\ell+1}=M-\ell\\k_1,\ldots,k_\ell}}c^{k_1\ldots k_\ell}_{j_1\ldots j_{\ell+1}}
	I^{j_1+j_{\ell+1}}\otimes b_{k_1}
	\otimes\ldots
	\otimes I^{j_\ell}\otimes b_{k_\ell}
	\\\nonumber
	&+\ldots
	+\sum_{k_1,\ldots,k_M}c^{k_1\ldots k_M}b_{k_1}\otimes\ldots\otimes b_{k_M}
	\\\label{Eq: all to the left-I representative}
	&=\sum_{j=0}^MI^{M-j}\otimes a'_j
	\in\bigoplus_{j=0}^MI^{M-j}\otimes B^j_{\textsc{irr}}\,.
\end{align}
We stress that some of the $a'_j$'s may vanish in the process.
However, it is important to observe that, moving from $a_M$ to $a_M'$, the $\gamma$-sequence (and thus its equivalence class) does not change.
Notice that, if we replace $a_M$ with $I^K\otimes a_M$ or $a_M\otimes I^K$, the $\widetilde{B}$-irreducible elements $\{a'_j\}_{j=0}^M$ do not change.
%Moreover, by Lemma \ref{Lem: PhiM is a bijection} for all $j\in\{0,\ldots,M\}$ there exists $\boldsymbol{\widetilde{a}}_j\in\boldsymbol{\widetilde{B}}^j$ such that $a'_j=\Phi_j(\boldsymbol{\widetilde{a}}_j)$.

\begin{definition}\label{Def: canonical representative of gamma-sequence}
	Let $(\overline{\gamma}^M_Na_M)_N$ be a $\gamma$-sequence and let $\sum_{j=0}^MI^{M-j}\otimes a'_j$ be the element defined as per Equation \eqref{Eq: all to the left-I representative}, where $a'_j\in B^j_{\textsc{irr}}$ for all $j\in\{0,\ldots,M\}$.
	The sequence
	\begin{align*}
		\sum_{j=0}^M(\overline{\gamma}^j_Na'_j)_N
		\in\dot{B}^\infty_\gamma\,,
	\end{align*}
	is called the \textbf{canonical representative} of $[\overline{\gamma}^M_Na_M]_N$.
\end{definition}

\begin{remark}\label{Rmk: the canonical representative is only asymptotically equivalent; it is enough to consider tildeB irreducible gamma-sequences}
	\noindent
	\begin{enumerate}[(i)]
		\item\label{Item: the canonical representative is only asymptotically equivalent}
		It is worth pointing out that, while $(\overline{\gamma}^M_Na_M)_N=(\overline{\gamma}^M_Na_M')_N$ for $a_M'=\sum_{j=0}^MI^{M-j}\otimes a'_j$, for the canonical representative we only have equality of equivalence classes, \textit{i.e.} $[\overline{\gamma}^M_Na_M]_N=\sum_{j=0}^M[\overline{\gamma}^j_Na'_j]_N$.
		In particular we have
		\begin{align}\label{Eq: relation between gamma sequence and its canonical representative}
			(\overline{\gamma}^M_Na_M)_N
			=(\overline{\gamma}^M_Na_M')_N
			=\sum_{j=0}^M(\overline{\gamma}^j_Na_j')_N
			+R_N\,,
		\end{align}
		where $\|R_N\|=O(1/N^\infty)$.
		For example if $a_M=a_0I^M$ then the canonical representative is the constant sequence $a_N=a_0I^N$, $N\in\mathbb{N}$, which coincides with $(\overline{\gamma}^M_Na_M)_N$ only for $N\geq M$.
		
		\item\label{Item: it is enough to consider tildeB irreducible gamma-sequences}
		On account of the previous discussion we observe that the algebra generated by $\gamma$-sequences of the form $(\overline{\gamma}^M_Na_M)_N$ for $a_M\in B^M_{\textsc{irr}}$, $M\in\mathbb{N}$, exhaust the whole space $\dot{B}^\infty_\gamma$.
	\end{enumerate}
\end{remark}

The following proposition shows that the canonical representative introduced in Definition \ref{Def: canonical representative of gamma-sequence} is unique.

\begin{proposition}\label{Prop: the canonical representative of a gamma-sequence is unique}
	Let $M\in\mathbb{N}$ and $a_j\in B^j_{\textsc{irr}}$ for all $j=0,\ldots, M$.
	Then
	\begin{align}\label{Eq: the canonical representative of a gamma-sequence is unique}
		\lim_{N\to\infty}\bigg\|\sum_{j=0}^M\overline{\gamma}^j_Na_j\bigg\|_N=0
		\Longleftrightarrow a_0=0\,,\ldots,a_M=0\,.
	\end{align}
\end{proposition}
\begin{proof}
	The proof is similar to the one of Lemma \ref{Lem: Btilde gamma-sequences retain information on the initial value}.
	By direct inspection we have
	\begin{align*}
		0=\lim_{N\to\infty}
		\bigg\|\sum_{j=0}^M\overline{\gamma}^j_Na_j\bigg\|_N
		\geq\lim_{N\to\infty}
		\bigg|\tau^N\bigg(\sum_{j=0}^M\overline{\gamma}^j_Na_j\bigg)\bigg|
		=|a_0|\,.
	\end{align*}
	Let now $\eta\in S(B)$ and let $\omega_{\eta,N}:=\tau^r\otimes(\tau^{2M-1}\otimes\eta)^q\in S(B^N)$, where $N=r+2Mq$, $q\in\mathbb{N}$ and $r\in\{0,\ldots, 2M-1\}$.
	We have
	\begin{align*}
		0=\lim_{N\to\infty}
		\bigg\|\sum_{j=0}^M\overline{\gamma}^j_Na_j\bigg\|_N
		\geq\lim_{N\to\infty}
		\bigg|\omega_{\eta,N}\bigg(\sum_{j=1}^M\overline{\gamma}^j_Na_j\bigg)\bigg|
		=\frac{1}{2M}|\eta(a_1)|\,,
	\end{align*}
	which implies $a_1=0$ because of the arbitrariness of $\eta\in S(B)$.
	Notice that $\omega_{\eta,N}(\overline{\gamma}_N^ja_j)=0$ for all $j\geq 2$ on account of the assumption $a_j\in B^j_{\textsc{irr}}$.
	
	Proceeding by induction we may assume that $a_1=\ldots=a_{\ell-1}=0$ and prove that $a_\ell=0$.
	To this avail let $\eta_1,\ldots,\eta_\ell\in S(B)$ and set
	\begin{align*}
		\omega_{\eta_1,\ldots,\eta_\ell,N}:=\tau^r\otimes
		(\tau^{2M-\ell}\otimes\eta_1\otimes\ldots\otimes\eta_\ell)^q
		\in S(B^N)\,,
	\end{align*}
	where $N=r+2Mq$, $q\in\mathbb{N}$ and $r\in\{0,\ldots,2M-1\}$.
	Using the inductive hypothesis we find
	\begin{align*}
		0=\lim_{N\to\infty}
		\bigg\|\sum_{j=\ell}^M\overline{\gamma}^j_Na_j\bigg\|_N
		\geq\lim_{N\to\infty}
		\bigg|\omega_{\eta_1,\ldots,\eta_\ell,N}\bigg(\sum_{j=\ell}^M\overline{\gamma}^j_Na_j\bigg)\bigg|
		=\frac{1}{2M}|(\eta_1\otimes\ldots\otimes\eta_\ell)(a_\ell)|\,,
	\end{align*}
	where, with the same argument as above, the contributions arising from $a_j$, $j\geq\ell+1$, vanish.
	The arbitrariness of $\eta_1,\ldots,\eta_\ell\in S(B)$ implies $a_\ell=0$.
\end{proof}

Summing up, every equivalence class $[\overline{\gamma}^M_Na_M]_N\in[\dot{B}]^\infty_\gamma$ has a unique canonical representative obtained by decomposing $a_M$ into its $\widetilde{B}$-irreducible components.

We shall now discuss the notion of canonical representative for a generic element $[a_N]_N\in[\dot{B}]^\infty_\gamma$.
Proposition \ref{Prop: the C* algebra of gamma-sequences is commutative} and Remark \ref{Rmk: the canonical representative is only asymptotically equivalent; it is enough to consider tildeB irreducible gamma-sequences}-\ref{Item: it is enough to consider tildeB irreducible gamma-sequences} lead to the following definition.

\begin{definition}\label{Def: canonical representative of product of gamma-sequences}
	Let $[a_N]_N\in[\dot{B}]^\infty_\gamma$ be such that
	\begin{multline*}
		[a_N]_N=\sum_{\ell,k_1,\ldots,k_\ell}c^{k_1\ldots k_\ell}
		[\overline{\gamma}^{M(k_1)}_N(a_{(k_1)})
		\cdots\overline{\gamma}^{M(k_\ell)}_N(a_{(k_\ell)})]_N
		\\=\sum_{\ell,k_1,\ldots,k_\ell}c^{k_1\ldots k_\ell}
		\frac{1}{N^{\ell-1}}
		\sum_{|j|_\ell=N-|M(k)|_\ell}
		[\overline{\gamma}_N(
		I^{j_1}\otimes
		\wick{
			\c1{a}_{(k_1)}\otimes
			\ldots\otimes I^{j_\ell}\otimes
			\c1{a}_{(k_\ell)}
		})]_N\,.
	\end{multline*}
	where $a_{k_j}\in B^{M(k_j)}_{\textsc{irr}}$ for all $k_j$, while the sum over $\ell,k_1,\ldots,k_\ell$ is finite and $|M(k)|_\ell:=M(k_1)+\ldots+M(k_\ell)$.
	The sequence
	\begin{align}
		\sum_{\ell,k_1,\ldots,k_\ell}c^{k_1\ldots k_\ell}
		\bigg(\frac{1}{N^{\ell-1}}
		\sum_{|j|_\ell=N-|M(k)|_\ell}
		\overline{\gamma}_N(
		I^{j_1}\otimes
		\wick{
			\c1{a}_{(k_1)}\otimes
			\ldots\otimes I^{j_\ell}\otimes
			\c1{a}_{(k_\ell)}
		})\bigg)_{N\geq |M(k)|_\ell}\,,
	\end{align}
	is called the \textbf{canonical representative} of $[a_N]_N$.
\end{definition}

Similarly to Proposition \ref{Prop: the canonical representative of a gamma-sequence is unique} we have the following result, showing that the canonical representative introduced in Definition \ref{Def: canonical representative of product of gamma-sequences} is unique.

\begin{proposition}\label{Prop: the canonical representative of a product of gamma-sequences is unique}
	It holds
	\begin{multline}\label{Eq: result for the canonical representative of a product of gamma-sequences is unique}
		\lim_{N\to\infty}\bigg\|
		\sum_{\ell,k_1,\ldots,k_\ell}c^{k_1\ldots k_\ell}
		\frac{1}{N^{\ell-1}}
		\sum_{|j|_\ell=N-|M(k)|_\ell}
		\overline{\gamma}_N(
		I^{j_1}\otimes
		\wick{
			\c1{a}_{(k_1)}\otimes
			\ldots\otimes I^{j_\ell}\otimes
			\c1{a}_{(k_\ell)}
		})
		\bigg\|_N=0
		\\\Longleftrightarrow
		\sum_{\ell,k_1,\ldots,k_\ell}c^{k_1\ldots k_\ell}
		\frac{1}{N^{\ell-1}}
		\sum_{|j|_\ell=N-|M(k)|_\ell}
		\overline{\gamma}_N(
		I^{j_1}\otimes
		\wick{
			\c1{a}_{(k_1)}\otimes
			\ldots\otimes I^{j_\ell}\otimes
			\c1{a}_{(k_\ell)}
		})=0
		\quad\forall N\in\mathbb{N}\,,
	\end{multline}
	where the sum over $\ell,k_1,\ldots,k_\ell\in\mathbb{N}$ is finite and $a_{k_j}\in B^{M(k_j)}_\textsc{irr}$ for all $k_j$.
\end{proposition}
\begin{proof}
	For the sake of clarity, we will discuss the proof for $\ell\leq 2$.
	This simplifies the construction without affecting the validity of the argument.
	We thus consider the sequence
	\begin{align}\label{Eq: sequence for the canonical representative of a product of gamma-sequences is unique}
		a_N:=\sum_{k_1,k_2}c^{k_1k_2}
		\frac{1}{N}\sum_{|j|=N-|M(k)|_2}
		\overline{\gamma}_N(I^{j_1}\otimes a_{(k_1)}\otimes I^{j_2}\otimes a_{(k_2)})\,,
	\end{align}
	where the sum over $k_1,k_2$ is finite and $a_{(k)}\in B^{M(k)}_{\textsc{irr}}$ for all $k$.
	Notice that, whenever $M(k_1)=0$ or $M(k_2)=0$ the corresponding contribution reduces to a single $\gamma$-sequence up to a remainder $O(1/N)$.
	We have to prove that $\lim\limits_{N\to\infty}\|a_N\|_N=0$ implies $a_N=0$ for all $N\in\mathbb{N}$.
	
	We observe that $\|a_N\|_N\underset{N\to\infty}{\longrightarrow}0$ entails
	\begin{align*}
		0=\lim_{N\to\infty}\|a_N\|_N\geq|\tau^N(a_N)|
		=\sum_{\substack{k_1\colon M(k_1)=0 \\ k_2\colon M(k_2)=0}}c^{k_1k_2}a_{(k_1)}a_{(k_2)}\,,
	\end{align*}
	so that we may assume $(M(k_1),M(k_2))\neq(0,0)$ in \eqref{Eq: sequence for the canonical representative of a product of gamma-sequences is unique}.
	
	We now analyse \eqref{Eq: sequence for the canonical representative of a product of gamma-sequences is unique} with the help of the following parameters:
	\begin{multline}\label{Eq: M-parameters}
		\overline{M}:=\max\limits_{k_1,k_2}\max\{M(k_1),M(k_2)\}\,,
		\\
		\underline{M}_1:=\min\limits_{k_1,k_2}\max\{M(k_1),M(k_2)\}\,,
		\qquad
		\underline{M}_2:=\min\limits_{\substack{k_1\colon M(k_1)\leq \underline{M}_1\\ k_2\colon M(k_2)\leq\underline{M}_1}}\min\{M(k_1),M(k_2)\}\,.
	\end{multline}
	Roughly speaking $\overline{M}$ is the maximal degree of the $a_{(k)}$'s appearing in \eqref{Eq: sequence for the canonical representative of a product of gamma-sequences is unique}.
	The parameter $\underline{M}_1\leq\overline{M}$ is the minimal "bigger length" among all pairs $(k_1,k_2)$ appearing in \eqref{Eq: sequence for the canonical representative of a product of gamma-sequences is unique}.
	Notice that $\underline{M}_1>0$ on account of the hypothesis $(M(k_1),M(k_2))\neq (0,0)$.
	Finally $\underline{M}_2\leq\underline{M}_1$ is the minimal length of the $a_{(k)}$'s appearing when considering only those pairs $(k_1,k_2)$ for which $\max\{M(k_1),M(k_2)\}\leq\underline{M}_1$ ---notice that this implies $M(k)=\underline{M}_1$ for at least one between $k\in\{k_1,k_2\}$.
	
	Let $\omega_{\underline{M}_1}\in S(B^{\underline{M}_1})$, $\omega_{\underline{M}_2}\in S(B^{\underline{M}_2})$ and let $\omega_{\underline{M}_1,\underline{M}_2,N}\in S(B^N)$ be defined by
	\begin{align}\label{Eq: state with trace desert and two states of minimal maximal length}
		\omega_{\underline{M}_1,\underline{M}_2,N}
		:=\tau^r\otimes
		(\tau^{\overline{M}}\otimes \omega_{\underline{M}_1}\otimes \tau^{\overline{M}}\otimes\omega_{\underline{M}_2})^q\,,
	\end{align}
	where $N=r+(2\overline{M}+\underline{M}_1+\underline{M}_2)q$, $q\in\mathbb{N}$, $r\in\{0,\ldots, 2\overline{M}+\underline{M}_1+\underline{M}_2-1\}$.
	We consider
	\begin{align}
		\nonumber
		\omega_{\underline{M}_1,\underline{M}_2,N}(a_N)
		&=\sum_{k_1,k_2}c^{k_1k_2}
		\frac{1}{N}\sum_{|j|=N-|M(k)|_2}
		\omega_{\underline{M}_1,\underline{M}_2,N}\bigg[\overline{\gamma}_N(I^{j_1}\otimes a_{(k_1)}\otimes I^{j_2}\otimes a_{(k_2)})\bigg]
		\\\label{Eq: evaluation of to be zero sequence on state with trace desert and two states with minimal maximal length}
		&=\sum_{\substack{k_1\colon M(k_1)\leq\underline{M}_1\\k_2\colon M(k_2)\leq\underline{M}_1}}c^{k_1k_2}
		\frac{1}{N}\sum_{|j|=N-|M(k)|_2}
		\omega_{\underline{M}_1,\underline{M}_2,N}\bigg[\overline{\gamma}_N(I^{j_1}\otimes a_{(k_1)}\otimes I^{j_2}\otimes a_{(k_2)})\bigg]\,,
	\end{align}
	where in the second line we observed that
	\begin{align*}
		\omega_{\underline{M}_1,\underline{M}_2,N}\bigg[\overline{\gamma}_N(I^{j_1}\otimes a_{(k_1)}\otimes I^{j_2}\otimes a_{(k_2)})\bigg]=0
		\qquad\textrm{if } \min\{M(k_1),M(k_2)\}>\underline{M}_1\,;
	\end{align*}
	This follows from the fact that, if $\underline{M}_1<M_2\leq\overline{M}$, for all $a_{M_2}\in B^{M_2}_{\textsc{irr}}$ we have
	\begin{align*}
		\omega_{\underline{M}_1,\underline{M}_2,N}(\overline{\gamma}_N(a_{N-M_2}\otimes a_{M_2}))=0\,,
	\end{align*}
	no matter the choice of $a_{N-M_2}\in B^{N-M_2}$.
	In fact, for all $j\in\{0,\ldots, N-1\}$ one finds that
	\begin{align*}
		\omega_{\underline{M}_1,\underline{M}_2,N}\big[
		\gamma_N^j(a_{N-M_2}\otimes a_{M_2})\big]
		=\tau^r\otimes
		(\tau^{\overline{M}}\otimes \omega_{\underline{M}_1}\otimes \tau^{\overline{M}}\otimes\omega_{\underline{M}_2})^q
		\big[
		\gamma_N^j(a_{N-M_2}\otimes a_{M_2})\big]\,,
	\end{align*}
	is non vanishing only if the position of $a_{M_2}$ "overlaps completely" with either $\omega_{\underline{M}_1}$ or with $\omega_{\underline{M}_2}$, however, this is not possible because $M_2>\underline{M}_1\geq\underline{M}_2$.
	Notice that overlapping with both states is impossible since each pair $\omega_{\underline{M}_1},\omega_{\underline{M}_2}$ is separated by $\tau^{\overline{M}}$ and $M_2\leq \overline{M}$.

	We now analyse the remaining contributions \eqref{Eq: evaluation of to be zero sequence on state with trace desert and two states with minimal maximal length} of $\omega_{\underline{M}_1,\underline{M}_2,N}(a_N)$.
	Notice that the condition $M(k_1)\leq\underline{M}_1$ and $M(k_2)\leq \underline{M}_1$ implies $M(k)=\underline{M}_1$ for at least one between $k_1,k_2$.
	In fact, we also have $M(k_1),M(k_2)\geq\underline{M}_2$ which implies  $M(k_2)=\underline{M}_2$ or $M(k_1)=\underline{M}_2$ for at least one pair $(k_1,k_2)$.
	Moreover, by direct inspection:
	\begin{enumerate}[(a)]
		\item\label{Item: case omegaM1omegaM2}
		If $j_2=\overline{M}\mod 2\overline{M}+\underline{M}_1+\underline{M}_2$ then
		\begin{multline}\label{Eq: application of ad hoc state for j2 equal to maximal length}
			\omega_{\underline{M}_1,\underline{M}_2,N}\big[\overline{\gamma}_N(I^{j_1}\otimes a_{(k_1)}\otimes I^{j_2}\otimes a_{(k_2)})\big]
			\\=\frac{1}{2\overline{M}+\underline{M}_1+\underline{M}_2}
			\bigg[\omega_{\underline{M}_1}(I^{\underline{M}_1-M(k_1)}\otimes a_{(k_1)})
			\omega_{\underline{M}_2}(I^{\underline{M}_2-M(k_2)}\otimes a_{(k_2)})
			\\+\omega_{\underline{M}_2}(I^{\underline{M}_2-M(k_1)}\otimes a_{(k_1)})
			\omega_{\underline{M}_1}(I^{\underline{M}_1-M(k_2)}\otimes a_{(k_2)})\bigg]
			+O(1/N)\,,
		\end{multline}
		with the convention that the contribution vanishes if, say, $\underline{M}_2<M(k_1)$ ---this may happen if $\underline{M}_2<\underline{M}_1$ and $M(k_1)=\underline{M}_1$.
		This restrict the non-vanishing contributions to those pairs $(k_1,k_2)$ such that $\{M(k_1),M(k_2)\}=\{\underline{M}_1,\underline{M}_2\}$.
		Notice that there exists at least one such pair on account of the definition of $\underline{M}_2$ ---\textit{cf.} Equation \eqref{Eq: M-parameters}.
		
		To prove \eqref{Eq: application of ad hoc state for j2 equal to maximal length} it suffices to observe that for all $\ell\in\{0,\ldots,N-1\}$ we have
		\begin{multline*}
			\omega_{\underline{M}_1,\underline{M}_2,N}\big[\gamma_N^\ell(I^{j_1}\otimes a_{(k_1)}\otimes I^{j_2}\otimes a_{(k_2)})\big]
			\\=\begin{dcases*}
				\omega_{\underline{M}_1}(I^{\underline{M}_1-M(k_1)}\otimes a_{(k_1)})
				\omega_{\underline{M}_2}(I^{\underline{M}_2-M(k_2)}\otimes a_{(k_2)})
				\\
				\textrm{if }
				\ell=0\mod \underline{M}_1+\underline{M}_2+2\overline{M}
				\\
				\omega_{\underline{M}_2}(I^{\underline{M}_2-M(k_1)}\otimes a_{(k_1)})
				\omega_{\underline{M}_1}(I^{\underline{M}_1-M(k_2)}\otimes a_{(k_2)})
				\\
				\textrm{if }
				\ell=\underline{M}_1+\overline{M}\mod \underline{M}_1+\underline{M}_2+2\overline{M}
				\\
				0\quad\textrm{otherwise}
			\end{dcases*}
		\end{multline*}
		Since the number of $\ell\in\{0,\ldots,N-1\}$ such that $\ell=0\mod 2\overline{M}+\underline{M}_1+\underline{M}_2$ (\textit{resp}. $\ell=\underline{M}_1+\overline{M}\mod 2\overline{M}+\underline{M}_1+\underline{M}_2$) is roughly $N/(2\overline{M}+\underline{M}_1+\underline{M}_2)+O(1)$ the formula for $\omega_{\underline{M}_1,\underline{M}_2,N}\big[\overline{\gamma}_N(I^{j_1}\otimes a_{(k_1)}\otimes I^{j_2}\otimes a_{(k_2)})\big]$ follows.
		
		\item\label{Item: case omegaM1omegaM1 omegaM2omegaM2}
		Similarly, if $j_2=2\overline{M}+\underline{M}_1\mod 2\overline{M}+\underline{M}_1+\underline{M}_2$ then
		\begin{multline*}
			\omega_{\underline{M}_1,\underline{M}_2,N}\big[\overline{\gamma}_N(I^{j_1}\otimes a_{(k_1)}\otimes I^{j_2}\otimes a_{(k_2)})\big]
			\\=\frac{1}{2\overline{M}+\underline{M}_1+\underline{M}_2}
			\bigg[\omega_{\underline{M}_1}(I^{\underline{M}_1-M(k_1)}\otimes a_{(k_1)})
			\omega_{\underline{M}_1}(I^{\underline{M}_1-M(k_2)}\otimes a_{(k_2)})
			\\+\omega_{\underline{M}_2}(I^{\underline{M}_2-M(k_1)}\otimes a_{(k_1)})
			\omega_{\underline{M}_2}(I^{\underline{M}_2-M(k_2)}\otimes a_{(k_2)})\bigg]
			+O(1/N)\,,
		\end{multline*}
		where again the contribution is non-vanishing if and only if $\{M(k_1),M(k_2)\}=\{\underline{M}_1,\underline{M}_2\}$.
		
		\item
		In all other cases the contribution vanishes.
	\end{enumerate}
	The number of $j_2\in\mathbb{N}$ such that $j_2\leq N-|\underline{M}(k)|_2$ and $j_2=\overline{M}\mod 2\overline{M}+\underline{M}_1+\underline{M}_2$ (\textit{resp.} $j_2=2\overline{M}+\underline{M}_1\mod 2\overline{M}+\underline{M}_1+\underline{M}_2$) is roughly $N/(2\overline{M}+\underline{M}_1+\underline{M}_2)+O(1)$.
	Moreover, we have
	\begin{align*}
		\omega_{\underline{M}_2}(I^{\underline{M}_2-M(k)}\otimes a_{(k)})
		=(\omega_{\underline{M}_1-\underline{M}_2}\otimes\omega_{\underline{M}_2})
		(I^{\underline{M}_1-M(k)}\otimes a_{(k)})\,,
	\end{align*}
	where $\omega_{\underline{M}_1-\underline{M}_2}\in S(B^{\underline{M}_1-\underline{M}_2})$ is arbitrarily chosen.
	
	Thus, combining cases \eqref{Item: case omegaM1omegaM2}-\eqref{Item: case omegaM1omegaM1 omegaM2omegaM2} we find
	\begin{multline*}
		\omega_{\underline{M}_1,\underline{M}_2,N}(a_N)=
		\sum_{\substack{(k_1,k_2)\colon\\\{M(k_1),M(k_2)\}=\{\underline{M}_1,\underline{M}_2\}}}
		\frac{c^{k_1k_2}}{(2\overline{M}+\underline{M}_1+\underline{M}_2)^2}
		\\\bigg(\frac{1}{2}\omega_{\underline{M}_1}
		+\frac{1}{2}\omega_{\underline{M}_1-\underline{M}_2}\otimes\omega_{\underline{M}_2}\bigg)^{2}
		\bigg[I^{\underline{M}_1-M(k_1)}\otimes a_{(k_1)}
		\otimes I^{\underline{M}_1-M(k_2)}\otimes a_{(k_2)}
		\bigg]
		+O(1/N)\,.
	\end{multline*}
	Since $\|a_N\|_N\geq |\omega_{\underline{M}_1,\underline{M}_2,N}(a_N)|$ and $\|a\|_N\underset{N\to\infty}{\longrightarrow}0$ we find
	\begin{align}\label{Eq: linear combination vanishing on square of sum of states}
		\sum_{\substack{(k_1,k_2)\colon\\\{M(k_1),M(k_2)\}=\{\underline{M}_1,\underline{M}_2\}}}
		\!\!\!\!\!\!\!\!\!\!\!\!\!\!\!\!\!\!\!\!\!
		c^{k_1k_2}
		\bigg(\frac{1}{2}\omega_{\underline{M}_1}
		+\frac{1}{2}\omega_{\underline{M}_1-\underline{M}_2}\otimes\omega_{\underline{M}_2}\bigg)^{2}
		\bigg[I^{\underline{M}_1-M(k_1)}\otimes a_{(k_1)}
		\otimes I^{\underline{M}_1-M(k_2)}\otimes a_{(k_2)}
		\bigg]=0\,,
	\end{align}
	for all $\omega_{\underline{M}_1}\in S(B^{\underline{M}_1})$, $\omega_{\underline{M}_2}\in S(B^{\underline{M}_2})$ and $\omega_{\underline{M}_1-\underline{M}_2}\in S(B^{\underline{M}_1-\underline{M}_2})$.
	Choosing
	\begin{align*}
		\omega_{\underline{M}_1}
		=\omega_{\underline{M}_1-\underline{M}_2}\otimes\omega_{\underline{M}_2}\,,
	\end{align*}
	we have
	\begin{align}\label{Eq: linear combination vanishing on square of sum of states - diagonal case for product states}
		\sum_{\substack{(k_1,k_2)\colon\\\{M(k_1),M(k_2)\}=\{\underline{M}_1,\underline{M}_2\}}}
		\!\!\!\!\!\!\!\!\!\!\!\!\!\!\!\!\!\!\!\!\!
		c^{k_1k_2}
		(\omega_{\underline{M}_1-\underline{M}_2}\otimes\omega_{\underline{M}_2})^2
		\bigg[I^{\underline{M}_1-M(k_1)}\otimes a_{(k_1)}
		\otimes I^{\underline{M}_1-M(k_2)}\otimes a_{(k_2)}\bigg]=0\,,
	\end{align}
	for all $\omega_{\underline{M}_2}\in S(B^{\underline{M}_2})$ and $\omega_{\underline{M}_1-\underline{M}_2}\in S(B^{\underline{M}_1-\underline{M}_2})$.
	This implies that in the general case, unfolding $(\omega_{\underline{M}_1}+\omega_{\underline{M}_1-\underline{M}_2}\otimes\omega_{\underline{M}_2})^2$ and using Equation \eqref{Eq: linear combination vanishing on square of sum of states - diagonal case for product states} we have
	\begin{multline}\label{Eq: linear combination vanishing on square of sum of states - linear identity for product states}
		\sum_{\substack{(k_1,k_2)\colon\\\{M(k_1),M(k_2)\}=\{\underline{M}_1,\underline{M}_2\}}}
		c^{k_1k_2}
		(\omega_{\underline{M}_1}\otimes \omega_{\underline{M}_1})
		\bigg[I^{\underline{M}_1-M(k_1)}\otimes a_{(k_1)}
		\otimes I^{\underline{M}_1-M(k_2)}\otimes a_{(k_2)}\bigg]
		\\+
		\sum_{\substack{(k_1,k_2)\colon\\\{M(k_1),M(k_2)\}=\{\underline{M}_1,\underline{M}_2\}}}
		\!\!\!\!\!\!\!\!\!\!\!\!\!\!\!\!\!\!\!\!\!
		c^{k_1k_2}
		\bigg(\omega_{\underline{M}_1}
		\otimes\omega_{\underline{M}_1-\underline{M}_2}\otimes\omega_{\underline{M}_2}\bigg)
		\bigg[(I^{\underline{M}_1-M(k_1)}\otimes a_{(k_1)})
		\otimes_\pi (I^{\underline{M}_1-M(k_2)}\otimes a_{(k_2)})
		\bigg]=0\,,
	\end{multline}
	for all $\omega_{\underline{M}_1}\in S(B^{\underline{M}_1})$, $\omega_{\underline{M}_2}\in S(B^{\underline{M}_2})$ and $\omega_{\underline{M}_1-\underline{M}_2}\in S(B^{\underline{M}_1-\underline{M}_2})$ while $a\otimes_\pi a':=a\otimes a'+a'\otimes a$.
	Equation \eqref{Eq: linear combination vanishing on square of sum of states - linear identity for product states} is now linear in $\omega_{\underline{M}_1-\underline{M}_2}\otimes\omega_{\underline{M}_2}$.
	Since convex combinations of states in $S(B^{\underline{M}_1-\underline{M}_2})\otimes S(B^{\underline{M}_2})$ generate $S(B^{\underline{M}_1})$ we find that
	\begin{multline}\label{Eq: linear combination vanishing on square of sum of states - different M1 states}
		\sum_{\substack{(k_1,k_2)\colon\\\{M(k_1),M(k_2)\}=\{\underline{M}_1,\underline{M}_2\}}}
		c^{k_1k_2}
		(\omega_{\underline{M}_1}\otimes\omega_{\underline{M}_1})
		\bigg[I^{\underline{M}_1-M(k_1)}\otimes a_{(k_1)}
		\otimes I^{\underline{M}_1-M(k_2)}\otimes a_{(k_2)}\bigg]
		\\+
		\sum_{\substack{(k_1,k_2)\colon\\\{M(k_1),M(k_2)\}=\{\underline{M}_1,\underline{M}_2\}}}
		c^{k_1k_2}
		\big(\omega_{\underline{M}_1}
		\otimes\omega_{\underline{M}_1}'\big)
		\bigg[(I^{\underline{M}_1-M(k_1)}\otimes a_{(k_1)})
		\otimes_\pi (I^{\underline{M}_1-M(k_2)}\otimes a_{(k_2)})
		\bigg]=0\,,
	\end{multline}
	for all $\omega_{\underline{M}_1},\omega_{\underline{M}_1}'\in S(B^{\underline{M}_1})$.
	Choosing $\omega_{\underline{M}_1}=\omega_{\underline{M}_1}'$ we find
	\begin{align*}\label{Eq: linear combination vanishing on square of sum of states - diagonal case}
		\sum_{\substack{(k_1,k_2)\colon\\\{M(k_1),M(k_2)\}=\{\underline{M}_1,\underline{M}_2\}}}
		c^{k_1k_2}
		(\omega_{\underline{M}_1}\otimes\omega_{\underline{M}_1})
		\bigg[I^{\underline{M}_1-M(k_1)}\otimes a_{(k_1)}
		\otimes I^{\underline{M}_1-M(k_2)}\otimes a_{(k_2)}\bigg]
		=0\,,
	\end{align*}
	out which Equation \eqref{Eq: linear combination vanishing on square of sum of states - different M1 states} simplifies to
	\begin{align}
		\sum_{\substack{(k_1,k_2)\colon\\\{M(k_1),M(k_2)\}=\{\underline{M}_1,\underline{M}_2\}}}
		c^{k_1k_2}
		\big(\omega_{\underline{M}_1}
		\otimes\omega_{\underline{M}_1}'\big)
		\bigg[(I^{\underline{M}_1-M(k_1)}\otimes a_{(k_1)})
		\otimes_\pi (I^{\underline{M}_1-M(k_2)}\otimes a_{(k_2)})
		\bigg]=0\,,
	\end{align}
	for all $\omega_{\underline{M}_1},\omega_{\underline{M}_1}'\in S(B^{\underline{M}_1})$.
	
	The arbitrariness of $\omega_{\underline{M}_1},\omega_{\underline{M}_1}'\in S(B^{\underline{M}_1})$ and the fact that any $\omega_{2\underline{M}_1}\in S(B^{2\underline{M}_1})$ can be written as a convex combination of product states in $\omega_{\underline{M}_1}\otimes\omega_{\underline{M}_1}'$ lead to
	\begin{align}\label{Eq: to be zero finite contribution}
		\sum_{\substack{(k_1,k_2)\colon\\\{M(k_1),M(k_2)\}=\{\underline{M}_1,\underline{M}_2\}}}
		c^{k_1k_2}
		(I^{\underline{M}_1-M(k_1)}\otimes a_{(k_1)})
		\otimes_\pi (I^{\underline{M}_1-M(k_2)}\otimes a_{(k_2)})
		=0\,.
	\end{align}
	On account of the symmetry in $k_1,k_2$ in the sum, we may assume that $M(k_1)=\underline{M}_1$ and $M(k_2)=\underline{M}_2$ for all pairs $(k_1,k_2)$.
	Equation \eqref{Eq: to be zero finite contribution} reduces to
	\begin{align}\label{Eq: to be zero finite contribution - simplified version}
		\sum_{\substack{k_1\colon M(k_1)=\underline{M}_1\\k_2\colon M(k_2)=\underline{M}_2}}
		c^{k_1k_2}
		a_{(k_1)}\otimes_\pi
		(I^{\underline{M}_1-\underline{M}_2}\otimes a_{(k_2)})
		=0\,.
	\end{align}
	Let $\omega_{\underline{M}_1}\in S(B^{\underline{M}_1})$ and $\omega_{\underline{M}_2}\in S(B^{\underline{M}_2})$ and let
	\begin{align*}
		\omega_{\underline{M}_1,\underline{M}_2}
		:=\frac{1}{2}\omega_{\underline{M}_1}\otimes\tau^{\underline{M}_1-\underline{M}_2}\otimes\omega_{\underline{M}_2}
		+\frac{1}{2}\tau^{\underline{M}_1-\underline{M}_2}\otimes\omega_{\underline{M}_2}\otimes\omega_{\underline{M}_1}\,.
	\end{align*}
	Then Equation \eqref{Eq: to be zero finite contribution - simplified version} leads to
	\begin{multline*}
		0=\omega_{\underline{M}_1,\underline{M}_2}
		\bigg(\sum_{\substack{k_1\colon M(k_1)=\underline{M}_1\\k_2\colon M(k_2)=\underline{M}_2}}
		c^{k_1k_2}
		a_{(k_1)}\otimes_\pi
		(I^{\underline{M}_1-\underline{M}_2}\otimes a_{(k_2)})\bigg)
		\\=(\omega_{\underline{M}_1}\otimes\omega_{\underline{M}_2})
		\bigg(\sum_{\substack{k_1\colon M(k_1)=\underline{M}_1\\k_2\colon M(k_2)=\underline{M}_2}}
		c^{k_1k_2}
		a_{(k_1)}\otimes_\pi
		a_{(k_2)})\bigg)\,,
	\end{multline*}
	with the convention that $\omega_{\underline{M}_\ell}(a_{(k)})=0$ if $M(k)\neq\underline{M}_\ell$.
	This shows that Equation \eqref{Eq: to be zero finite contribution - simplified version} implies
	\begin{align}\label{Eq: to be zero finite contribution - simplified version without identities}
		\sum_{\substack{k_1\colon M(k_1)=\underline{M}_1\\k_2\colon M(k_2)=\underline{M}_2}}
		c^{k_1k_2}
		a_{(k_1)}\otimes_\pi
		\otimes a_{(k_2)}
		=0\,.
	\end{align}
	We now observe that Equation \eqref{Eq: to be zero finite contribution - simplified version without identities} is equivalent to
	\begin{align}\label{Eq: to be zero finite contribution in gamma-sequence}
		\sum_{\substack{k_1\colon M(k_1)=\underline{M}_1\\k_2\colon M(k_2)=\underline{M}_2}}
		c^{k_1k_2}
		\frac{1}{N}\sum_{j_1+j_2=N-\underline{M}_1-\underline{M}_2}
		\overline{\gamma}_N\big(I^{j_1}
		\otimes
		\wick{
		\c1
		a_{(k_1)}
		\otimes
		I^{j_2}\otimes
		\c1 a_{(k_2)}})
		=0\,,
		\qquad
		\forall N\geq \underline{M}_1-\underline{M}_2\,.
	\end{align}
	Indeed, by direct inspection Equation \eqref{Eq: to be zero finite contribution - simplified version without identities} implies that
	\begin{align}
		\sum_{\substack{k_1\colon M(k_1)=\underline{M}_1\\k_2\colon M(k_2)=\underline{M}_2}}
		c^{k_1k_2}
		I^{j_1}\otimes
		\wick{
		\c1
		a_{(k_1)}\otimes
		I^{j_2}\otimes
		\c1
		a_{(k_2)}
		}
		=0\,,
		\qquad\forall j_1,j_2\in\mathbb{N}\,,
	\end{align}
	and thus it implies Equation \eqref{Eq: to be zero finite contribution in gamma-sequence}.
	Conversely, if Equation \eqref{Eq: to be zero finite contribution in gamma-sequence} holds true then evaluation on the state $\tau^{\ell_1}\otimes\omega_{\underline{M}_1}\otimes\tau^{\ell_2}\otimes\omega_{\underline{M}_2}$ leads to
	\begin{multline*}
		0=(\tau^{\ell_1}\otimes\omega_{\underline{M}_1}
		\otimes\tau^{\ell_2}\otimes\omega_{\underline{M}_2})
		\bigg(
		\sum_{\substack{k_1\colon M(k_1)=\underline{M}_1\\k_2\colon M(k_2)=\underline{M}_2}}
		c^{k_1k_2}
		\frac{1}{N}\sum_{j_1+j_2=N-\underline{M}_1-\underline{M}_2}
		\overline{\gamma}_N\big(I^{j_1}
		\otimes
		\wick{
		\c1 a_{(k_1)}
		\otimes
		I^{j_2}\otimes
		\c1 a_{(k_2)}})
		\bigg)
		\\=\frac{1}{N}(\omega_{\underline{M}_1}\otimes\omega_{\underline{M}_2})
		\bigg(
		\sum_{\substack{k_1\colon M(k_1)=\underline{M}_1\\k_2\colon M(k_2)=\underline{M}_2}}
		c^{k_1k_2}
		\frac{1}{N} a_{(k_1)}
		\otimes_\pi a_{(k_2)}
		\bigg)\,,
	\end{multline*}
	where $\ell_1,\ell_2$ are such that $\ell_1+\ell_2=N-\underline{M}_1-\underline{M}_2$ while $\omega_{\underline{M}_1}\in S(B^{\underline{M}_1})$ and $\omega_{\underline{M}_2}\in S(B^{\underline{M}_2})$ are arbitrary states.
	This implies Equation \eqref{Eq: to be zero finite contribution - simplified version}.
	
	By comparison with \eqref{Eq: sequence for the canonical representative of a product of gamma-sequences is unique} we conclude that Equation \eqref{Eq: to be zero finite contribution in gamma-sequence} is nothing but the sum of the terms in $a_N$ whose pairs $k_1,k_2$ fulfils $\{M(k_1),M(k_2)\}=\{\underline{M}_1,\underline{M}_2\}$.
	
	At this stage we may either argue that this is in contradiction with the definition of $\underline{M}_1,\underline{M}_2$ ---unless $a_N=0$--- because $\min\limits_{\substack{k_1\colon M(k_1)\leq\underline{M}_1\\k_2\colon M(k_2)\leq\underline{M}_1}}\min\{M(k_1),M(k_2)\}>\underline{M}_2$.
	Alternatively we may consider the remaining contribution to $a_N$ and argue again as above identifying new values $\underline{M}_1,\underline{M}_2,\overline{M}$.
	In either case we have $a_N=0$ for all $N\in\mathbb{N}$ as claimed.
\end{proof}

\begin{remark}\label{Rmk: canonical representative for symmetric sequences}
	The notion of canonical representative applies also for symmetric sequences.
	Indeed, let $[\pi^M_Na_M]_N\in [B]^\infty_\pi$ where $a_M\in B^M_\pi$.
	As for $\gamma$-sequences one has \\ $[\pi^M_Na_M]_N=[\pi^{M+K}_N(I^K\otimes a_M)]_N$ so that the $\pi$-sequence generating $[\pi^M_Na_M]_N$ is not uniquely determined.
	Nevertheless, since $a_M\in B^M_\pi$, one obtain the following unique decomposition:
	\begin{align*}
		a_M=S_M(\widetilde{a}_0I^M+I^{M-1}\otimes \widetilde{a}_1+\ldots+\widetilde{a}_M)
		=S_M\bigg(\sum_{j=0}^MI^{M-j}\otimes\widetilde{a}_j\bigg)
		\in S_M\bigg(\bigoplus_{j=0}^M I^{M-j}\otimes\widetilde{B}^j_\pi\bigg)\,.
	\end{align*}
	With this decomposition at hand the canonical representative of $[\pi^M_Na_M]_N$ is defined by
	\begin{align*}
		\bigg(\sum_{j=0}^M\pi^j_N\widetilde{a}_j\bigg)_N\,.
	\end{align*}
	This point of view is equivalent to the one adopted in \cite{Landsman_Moretti_vandeVen_2020}.
\end{remark}

\subsection{The Poisson structure of $[B]^\infty_\gamma$}
\label{Subsec: the Poisson structure of the algebra of gamma-sequences}

In this section we will endow $[B]^\infty_\gamma$ with a Poisson structure defined on $[\dot{B}]^\infty_\gamma$.
Eventually we will discuss the deformation quantization of $[B]_\gamma$.

We recall that a \textbf{Poisson structure} over a $C^*$-algebra $\mathcal{A}$ is given by a bilinear map $\{\;,\;\}\colon \mathcal{A}_0\times \mathcal{A}_0\to\mathcal{A}_0$ defined on a dense $*$-subalgebra $\mathcal{A}_0\subset\mathcal{A}$ which fulfils:
\begin{align}
	\label{Eq: skew-symmetry of Poisson structure}
	\{a,a'\}&=-\{a',a\}\,,
	\qquad
	\{a,a'\}^*=\{a^*,a^{\prime *}\}
	\\
	\label{Eq: derivation property of Poisson structure}
	\{a,a'a''\}&=\{a,a'\}a''+a'\{a,a''\}\,,
	\\
	\label{Eq: Jacobi identity of Poisson structure}
	\{a,\{a',a''\}\}
	&=\{\{a,a'\},a''\}
	+\{a',\{a,a''\}\}\,,
\end{align}
for all $a,a',a''\in \mathcal{A}_0$.

\begin{proposition}\label{Prop: Poisson structure on algebra of gamma-sequences}
	Let $\{\;,\;\}_\gamma\colon[\dot{B}]^\infty_\gamma\times[\dot{B}]^\infty_\gamma\to[\dot{B}]^\infty_\gamma$ be the bilinear map defined by
	\begin{align}\label{Eq: Poisson structure on equivalence classes of gamma-sequences}
		\{[a_N]_N,[a_N']_N\}_\gamma
		:=[iN[a_N^{\textsc{can}},a_N^{\prime\,\textsc{can}}]]_N\,,
	\end{align}
	where $(a^{\textsc{can}}_N)_N$ denotes the canonical representative of $[a_N]_N$ ---\textit{cf.} Definitions \ref{Def: canonical representative of gamma-sequence}-\ref{Def: canonical representative of product of gamma-sequences}.
	Then $\{\;,\;\}_\gamma$ is a Poisson structure on $[B]^\infty_\gamma$.
\end{proposition}
\begin{proof}
	Notice that $\{\;,\;\}_\gamma$ fulfils condition \eqref{Eq: skew-symmetry of Poisson structure} because so does the pointwise commmutator $i[\;,\;]$.
	
	The non-trivial part of the proof is to prove that $\{\;,\;\}_\gamma$ is well-defined, namely that $\{[a_N]_N,[a_N']_N\}_\gamma$ is a well-defined element of $[\dot{B}]^\infty_\gamma$.
	Moreover, we also have to prove conditions \eqref{Eq: derivation property of Poisson structure}-\eqref{Eq: Jacobi identity of Poisson structure}: The latter do not follow from the properties of the commutator because Equation \eqref{Eq: Poisson structure on equivalence classes of gamma-sequences} uses the canonical representative and in general
	$[a_N^{\textsc{can}},a^{\textsc{can}}_N]\neq [a_N^{\textsc{can}},a^{\textsc{can}}_N]^{\textsc{can}}$.
	For these reasons we proceed in several steps:
	
	\begin{description}
		\item[$\boxed{\gamma}$]
		As a first step, we prove that $\{[a_N]_N,[a_N']_N\}_\gamma\in[\dot{B}]^\infty_\gamma$ for the case of two equivalence classes of $\gamma$-sequences.
		Since the commutator is linear, on account of Remark \ref{Rmk: the canonical representative is only asymptotically equivalent; it is enough to consider tildeB irreducible gamma-sequences}-\ref{Item: it is enough to consider tildeB irreducible gamma-sequences} we may reduce to the case $[a_N]_N=[\overline{\gamma}_N^Ma_M]_N$, $[a_N']_N=[\overline{\gamma}_N^{M'}a_{M'}]_N$ for $a_M\in B^M_{\textsc{irr}}$ and $a_{M'}\in B^{M'}_{\textsc{irr}}$.
		
		In this latter case we find, for large enough $N$, say $N\geq 2(M+M')$,
		\begin{multline*}
			iN[\overline{\gamma}^M_Na_M,\overline{\gamma}^{M'}_Na_{M'}]
			=iN\overline{\gamma}_N\bigg(
			\bigg[I^{N-M}\otimes a_M,\overline{\gamma}_N(I^{N-M'}\otimes a_{M'})
			\bigg]
			\bigg)
			\\=i\overline{\gamma}_N\bigg(
			\sum_{\substack{j\in\{0,\ldots,M-1\}\\\cup\{N-M'-1,\ldots,N-1\}}}\bigg[I^{N-M}\otimes a_M,\gamma_N^j(I^{N-M'}\otimes a_{M'})
			\bigg]
			\bigg)
			\\=i\overline{\gamma}_N\bigg(I^{N-M-2M'}\otimes
			\sum_{j=0}^{M+M'}\bigg[I^{M'}\otimes a_M\otimes I^{M'},\gamma_{M+2M'}^j(I^{M'+M}\otimes a_{M'})
			\bigg]
			\bigg)
			\\=:\overline{\gamma}_N^{M+2M'}a_{M+2M'}\,,
		\end{multline*}
		which implies
		\begin{align*}
			\{[\overline{\gamma}^M_Na_M]_N,[\overline{\gamma}^{M'}_Na_{M'}]_N\}_\gamma
			=[iN[\overline{\gamma}^M_Na_M,\overline{\gamma}^{M'}_Na_{M'}]]_N
			=[\overline{\gamma}^{M+2M'}_Na_{M+2M'}]_N
			\in[\dot{B}]^\infty_\gamma\,.
		\end{align*}
		This proves that the Poisson bracket between $[\overline{\gamma}^M_Na_M]_N$ and $[\overline{\gamma}^{M'}_Na_{M'}]_N$ is an element of $[\dot{B}]^\infty_\gamma$.
		
		\item[$\boxed{[\dot{B}]^\infty_\gamma}$]
		We now consider the general case of $[a_N]_N,[a_N']_N\in[\dot{B}]^\infty_\gamma$.
		Using again linearity and Remark \ref{Rmk: the canonical representative is only asymptotically equivalent; it is enough to consider tildeB irreducible gamma-sequences}-\ref{Item: it is enough to consider tildeB irreducible gamma-sequences} we may restrict to case 
		\begin{align*}
			[a_N]_N=
			[\overline{\gamma}_N^{M_1}a_{M_1}\cdots
			\overline{\gamma}_N^{M_\ell}a_{M_\ell}]_N\,,
			\qquad
			[a_N']_N=
			[\overline{\gamma}_N^{M_1'}a_{M_1'}\cdots
			\overline{\gamma}_N^{M_{\ell'}'}a_{M_{\ell'}'}]_N\,,
		\end{align*}
		where $a_M,a_M'\in B^{M}_{\textsc{irr}}$ for all $M$'s and $\ell,\ell'\in\mathbb{N}$ are arbitrary but fixed.
		
		We observe that Equation \eqref{Eq: product of gamma sequences - estimate of remainder} ---\textit{cf.} Proposition \ref{Prop: the C* algebra of gamma-sequences is commutative}--- leads to
		\begin{align*}
			(a_N^{\textsc{can}})_N
			=(\overline{\gamma}_N^{M_1}a_{M_1}\cdots
			\overline{\gamma}_N^{M_\ell}a_{M_\ell})_N
			+R_N\,,
			\qquad
			(a_N^{\prime\textsc{can}})_N
			=(\overline{\gamma}_N^{M_1'}a_{M_1'}\cdots
			\overline{\gamma}_N^{M_{\ell'}'}a_{M_{\ell'}'})_N
			+R_N'\,,
		\end{align*}
		where $\|R_N\|_N=O(1/N)=\|R_N'\|_N$.
		This implies
		\begin{align*}
			[a_N^{\textsc{can}},a^{\prime\textsc{can}}_N]
			&=[\overline{\gamma}_N^{M_1}a_{M_1}\cdots
			\overline{\gamma}_N^{M_\ell}a_{M_\ell},
			\overline{\gamma}_N^{M_1'}a_{M_1'}\cdots
			\overline{\gamma}_N^{M_{\ell'}'}a_{M_{\ell'}'}
			]
			\\&+[R_N,\overline{\gamma}_N^{M_1'}a_{M_1'}\cdots
			\overline{\gamma}_N^{M_{\ell'}'}a_{M_{\ell'}'}]
			+[\overline{\gamma}_N^{M_1}a_{M_1}\cdots
			\overline{\gamma}_N^{M_\ell}a_{M_\ell},R_N']
			\\&+[R_N,R_N']\,.
		\end{align*}
		At this stage we observe that
		\begin{align*}
			\big[iN[\overline{\gamma}_N^{M_1}a_{M_1}\cdots
			\overline{\gamma}_N^{M_\ell}a_{M_\ell},
			\overline{\gamma}_N^{M_1'}a_{M_1'}\cdots
			\overline{\gamma}_N^{M_{\ell'}'}a_{M_{\ell'}'}
			]\big]_N\in [\dot{B}]^\infty_\gamma\,.
		\end{align*}
		Indeed, this is due to the identity
		\begin{align*}
			[a_N,a_N'a_N'']
			=[a_N,a_N']a_N''
			+a_N'[a_N,a_N'']\,,
		\end{align*}
		together with the fact that the result holds true for $\ell=\ell'=1$.
		Moreover, we have $\|N[R_N,R_N']\|_N=O(1/N)$ so that it remains to discuss the term
		\begin{align*}
			\big\|N[R_N,\overline{\gamma}_N^{M_1'}a_{M_1'}\cdots
			\overline{\gamma}_N^{M_{\ell'}'}a_{M_{\ell'}'}]\big\|_N
			=O(1/N)\,,
		\end{align*}
		where the latter estimate is due to Remark \ref{Rmk: estimate of commutator with the remainder; Cstar algebra of symmetric sequences; character space for algebra of gamma sequences}-\ref{Item: estimate of commutator with the remainder}.
		Overall we have shown that
		\begin{align*}
			iN[a_N^{\textsc{can}},a^{\prime\textsc{can}}_N]
			=iN[\overline{\gamma}_N^{M_1}a_{M_1}\cdots
			\overline{\gamma}_N^{M_\ell}a_{M_\ell},
			\overline{\gamma}_N^{M_1'}a_{M_1'}\cdots
			\overline{\gamma}_N^{M_{\ell'}'}a_{M_{\ell'}'}]
			+R_N''\,,
		\end{align*}
		where $\|R_N''\|_N=O(1/N)$ and by direct inspection fulfils Equation \eqref{Eq: estimate on commutator with the remainder}.
		This implies in particular that
		\begin{align*}
			\big[iN[a_N^{\textsc{can}},a^{\prime\textsc{can}}_N]\big]_N
			=\big[iN[\overline{\gamma}_N^{M_1}a_{M_1}\cdots
			\overline{\gamma}_N^{M_\ell}a_{M_\ell},
			\overline{\gamma}_N^{M_1'}a_{M_1'}\cdots
			\overline{\gamma}_N^{M_{\ell'}'}a_{M_{\ell'}'}
			]\big]_N\in[\dot{B}]^\infty_\gamma\,.
		\end{align*}
		so that $\{\;,\;\}_\gamma$ is well-defined.
		
		\item[$\boxed{\eqref{Eq: derivation property of Poisson structure}-\eqref{Eq: Jacobi identity of Poisson structure}}$]
		By proceeding in a completely analogous way one also proves conditions \eqref{Eq: derivation property of Poisson structure}-\eqref{Eq: Jacobi identity of Poisson structure}.
		Indeed, considering without loss of generality
		\begin{multline*}
			[a_N]_N=
			[\overline{\gamma}_N^{M_1}a_{M_1}\cdots
			\overline{\gamma}_N^{M_\ell}a_{M_\ell}]_N\,,
			\qquad
			[a_N']_N=
			[\overline{\gamma}_N^{M_1'}a_{M_1'}\cdots
			\overline{\gamma}_N^{M_{\ell'}'}a_{M_{\ell'}'}]_N\,,
			\\
			[a_N'']_N=
			[\overline{\gamma}_N^{M_1''}a_{M_1''}\cdots
			\overline{\gamma}_N^{M_{\ell''}''}a_{M_{\ell''}''}]_N\,,
		\end{multline*}
		we find
		\begin{multline*}
			(iN)^2[a_N^{\textsc{can}},[a_N^{\prime\textsc{can}},a_N^{\prime\prime\textsc{can}}]]
			\\=iN[a_N^{\textsc{can}},
			iN[\overline{\gamma}_N^{M_1'}a_{M_1'}\cdots
			\overline{\gamma}_N^{M_{\ell'}'}a_{M_{\ell'}'},
			\overline{\gamma}_N^{M_1''}a_{M_1''}\cdots
			\overline{\gamma}_N^{M_{\ell''}''}a_{M_{\ell''}''}
			]]
			+iN[a_N^{\textsc{can}},R_N]
			\\=iN[\overline{\gamma}_N^{M_1}a_{M_1}\cdots
			\overline{\gamma}_N^{M_\ell}a_{M_\ell},
			iN[\overline{\gamma}_N^{M_1'}a_{M_1'}\cdots
			\overline{\gamma}_N^{M_{\ell'}'}a_{M_{\ell'}'},
			\overline{\gamma}_N^{M_1''}a_{M_1''}\cdots
			\overline{\gamma}_N^{M_{\ell''}''}a_{M_{\ell''}''}
			]]
			\\+iN[R_N',iN[\overline{\gamma}_N^{M_1'}a_{M_1'}\cdots
			\overline{\gamma}_N^{M_{\ell'}'}a_{M_{\ell'}'},
			\overline{\gamma}_N^{M_1''}a_{M_1''}\cdots
			\overline{\gamma}_N^{M_{\ell''}''}a_{M_{\ell''}''}
			]]
			+iN[a_N^{\textsc{can}},R_N]\,.
		\end{multline*}
		The first contribution fulfils \eqref{Eq: Jacobi identity of Poisson structure}.
		With an argument similar in spirit to Remark \ref{Rmk: estimate of commutator with the remainder; Cstar algebra of symmetric sequences; character space for algebra of gamma sequences}-\ref{Item: estimate of commutator with the remainder}, the second contribution can be estimated by
		\begin{align*}
			\|[R_N',iN[\overline{\gamma}_N^{M_1'}a_{M_1'}\cdots
			\overline{\gamma}_N^{M_{\ell'}'}a_{M_{\ell'}'},
			\overline{\gamma}_N^{M_1''}a_{M_1''}\cdots
			\overline{\gamma}_N^{M_{\ell''}''}a_{M_{\ell''}''}
			]]\|_N=O(1/N^2)\,.
		\end{align*}
		Finally $\|N[a_N^{\textsc{can}},R_N]\|_N=O(1/N)$ because of Remark \ref{Rmk: estimate of commutator with the remainder; Cstar algebra of symmetric sequences; character space for algebra of gamma sequences}-\ref{Item: estimate of commutator with the remainder} so that
		\begin{multline*}
			(iN)^2[a_N^{\textsc{can}},[a_N^{\prime\textsc{can}},a_N^{\prime\prime\textsc{can}}]]
			\\=iN[\overline{\gamma}_N^{M_1}a_{M_1}\cdots
			\overline{\gamma}_N^{M_\ell}a_{M_\ell},
			iN[\overline{\gamma}_N^{M_1'}a_{M_1'}\cdots
			\overline{\gamma}_N^{M_{\ell'}'}a_{M_{\ell'}'},
			\overline{\gamma}_N^{M_1''}a_{M_1''}\cdots
			\overline{\gamma}_N^{M_{\ell''}''}a_{M_{\ell''}''}
			]]+R_N'''\,,
		\end{multline*}
		with $\|R_N'''\|_N=O(1/N)$.
		This proves condition \ref{Eq: Jacobi identity of Poisson structure} for $\{\;,\;\}_\gamma$.
		
		By proceeding in a similar fashion we also have
		\begin{multline*}
		    iN[a_N^{\textsc{can}},a^{\prime\textsc{can}}_Na^{\prime\prime\textsc{can}}_N]
		    \\=iN[\overline{\gamma}_N^{M_1}a_{M_1}\cdots
			\overline{\gamma}_N^{M_\ell}a_{M_\ell},
			\overline{\gamma}_N^{M_1'}a_{M_1'}\cdots
			\overline{\gamma}_N^{M_{\ell'}'}a_{M_{\ell'}'}
			\overline{\gamma}_N^{M_1''}a_{M_1''}\cdots
			\overline{\gamma}_N^{M_{\ell''}''}a_{M_{\ell''}''}
			]
			+R_N'''\,,
		\end{multline*}
		 where $\|R_N'''\|_N=O(1/N)$ while the first contribution fulfils \eqref{Eq: derivation property of Poisson structure}.
		 This proves condition \eqref{Eq: derivation property of Poisson structure} for $\{\;,\;\}_\gamma$.
	\end{description}	
\end{proof}

\begin{remark}\label{Rmk: canonical representative for Poisson bracket}
	The proof of Proposition \ref{Prop: Poisson structure on algebra of gamma-sequences} shows that, if $a_M\in B^M_{\textsc{irr}}$ and $a_{M'}\in B^{M'}_{\textsc{irr}}$ then
	\begin{align*}
		\{[\overline{\gamma}^M_Na_M]_N,[\overline{\gamma}^{M'}_Na_{M'}]_N\}_\gamma
		=[\overline{\gamma}^{M+2M'}_Na_{M+2M'}]_N\,,
	\end{align*}
	where $a_{M+2M'}\in B^{M+2M'}$ is not $\widetilde{B}$-irreducible in general.
\end{remark}

At last, we can finally state and prove the main theorem of this paper.

\begin{theorem}\label{Thm: deformation quantization of the algebra of gamma sequences}
	Let $[B]_\gamma\subset\prod_{N\in\overline{\mathbb{N}}}[B]^N_\gamma$ be the continuous bundle of $C^*$-algebras defined as per Proposition \ref{Prop: continuous bundle of Cstar algebra for gamma-sequences}.
	For $K\in\overline{\mathbb{N}}$ let $Q_K\colon [\dot{B}]^\infty_\gamma\to [B]^K_\gamma$ be the linear map defined by
	\begin{align}\label{Eq: quantization maps}
		Q_K([a_N]_N):=
		\begin{cases}
			a_K^{\textsc{can}}
			& K\in\mathbb{N}
			\\
			[a_N]_N
			& K=\infty
		\end{cases}
	\end{align}
	where $(a_N^{\textsc{can}})_N$ is the canonical representative of $[a_N]_N$ as per Definitions \ref{Def: canonical representative of gamma-sequence}-\ref{Def: canonical representative of product of gamma-sequences}.
	Then the family of maps $\{Q_N\}_{N\in\overline{\mathbb{N}}}$ defines a strict deformation quantization of $[B]^\infty_\gamma$.
\end{theorem}
\begin{proof}
	Notice that $Q_N$ is well-defined for all $N\in\overline{\mathbb{N}}$ on account of the uniqueness of the canonical representative ---\textit{cf.} Propositions \ref{Prop: the canonical representative of a gamma-sequence is unique}-\ref{Prop: the canonical representative of a product of gamma-sequences is unique}.
	
	With reference to Section \ref{Sec: Introduction} we have
	\begin{align*}
		[B]_\gamma\leftrightarrow\prod_{N\in\overline{\mathbb{N}}}\mathcal{A}_N\,,
		\qquad
		[B]^N_\gamma\leftrightarrow\mathcal{A}_N\,,
		\qquad
		[\dot{B}]^\infty_\gamma\leftrightarrow\widetilde{\mathcal{A}}_\infty\,,
		\qquad
		[B]^\infty_\gamma\leftrightarrow\mathcal{A}_\infty\,.
	\end{align*}
	We will now prove conditions \ref{Item: quantization maps are Hermitian and define a continuous section}-\ref{Item: quantization maps fulfils the DGR condition}-\ref{Item: quantization maps are strict}.
	
	\begin{description}
		\item[$\boxed{\ref{Item: quantization maps are Hermitian and define a continuous section}}$]
		Per definition we have $Q_\infty:=\operatorname{Id}_{[\dot{B}]^\infty_\gamma}$ as well as $Q_N([a_N]_N)^*=Q_N([a_N^*]_N)$ for all $[a_N]_N\in[\dot{B}]^\infty_\gamma$.
		Moreover, Equation \eqref{Eq: quantization maps} defines an element in the space $[\dot{B}]_\gamma$ ---\textit{cf.} Equation \eqref{Eq: a posteriori continuous sections of the gamma-bundle}--- and thus a continuous section of $[B]_\gamma$ as per Proposition \ref{Prop: continuous bundle of Cstar algebra for gamma-sequences}.
		
		\item[$\boxed{\ref{Item: quantization maps fulfils the DGR condition}}$]
		By direct inspection one has
		\begin{align*}
			[iK[Q_K([a_N]_N),Q_K([a_N']_N)]]_K
			=[iK[a_K^{\textsc{can}},a_K^{\prime\textsc{can}}]]_K
			=\{[a_N]_N,[a_N']_N\}_\gamma\,,
		\end{align*}
		which implies Equation \eqref{Eq: Dirac-Groenewold-Rieffel condition}.
		Notice that, on account of Remark \ref{Rmk: canonical representative for Poisson bracket}, in general
		\begin{align*}
			Q_K(\{[a_N]_N,[a_N']_N\}_\gamma)
			\neq iK[a_K^{\textsc{can}},a_K^{\prime\textsc{can}}]\,,
		\end{align*}
		despite the fact that the equivalence classes of the associated sequences are equal.
		
		\item[$\boxed{\ref{Item: quantization maps are strict}}$]
		By direct inspection one finds that $Q_M([\dot{B}]^\infty_\gamma)=B^M_\gamma$ for all $M\in\mathbb{N}$.
		Indeed, by proceeding as in the proof of Proposition \ref{Prop: continuous bundle of Cstar algebra for gamma-sequences}, let $a_M\in B^M_\gamma$, for $M\in\mathbb{N}$.
		Then we have $a_M=\overline{\gamma}_M(a_M)=\sum_{j=0}^M\overline{\gamma}_M^ja_j$ for $a_j\in B^j_{\textsc{irr}}$ ---\textit{cf.} Equation \eqref{Eq: all to the left-I representative}.
		This implies that
		\begin{align*}
		    a_M
		    =\overline{\gamma}_M(a_M)
		    =\sum_{j=0}^M\overline{\gamma}_M^ja_j
		    =Q_M\bigg(\bigg[
		    \sum_{j=0}^M\overline{\gamma}_N^ja_j
		    \bigg]_N
		    \bigg)\,,
		\end{align*}
		 thus proving that $Q_M([\dot{B}]^\infty_\gamma)=B^M_\gamma$.
	\end{description}
\end{proof}

\appendix
\section{Characterization of $\widetilde{B}$-irreducible elements}
\label{App: characterization of tildeB irredubible elements}

This section is devoted to characterize the set $B^M_{\textsc{irr}}$ of $\widetilde{B}$-irreducible elements in $B^M$ ---\textit{cf.} Definition \ref{Def: Btilde-irreducible element}.
To this avail, we introduce the following convenient family of linear maps.
\begin{definition}\label{Def: Btildeell to BM embedding map}
	Let $M,\ell\in\mathbb{N}$, $\ell\geq 2$, and $j_1,\ldots,j_{\ell-1}\in\mathbb{N}$ such that $j_1+\ldots+j_{\ell-1}=M-\ell$.
	We denote by $\iota^{j_1\ldots j_{\ell-1}}_M\colon \widetilde{B}^\ell\to B^M$ the linear map defined by
	\begin{align}\label{Eq: Btildeell to BM embedding map}
		\iota^{j_1\ldots j_{\ell-1}}_M(\widetilde{a}_\ell)
		:=\sum_{k_1,\ldots, k_\ell}
		c^{k_1\ldots k_\ell}
		b_{k_1}\otimes I^{j_1}\otimes\ldots\otimes I^{j_{\ell-1}}\otimes b_{k_\ell}\,,
	\end{align}
	where $I,b_1,\ldots,b_{\kappa^2-1}$ is a basis of $B$ fulfilling \eqref{Eq: Mk-basis properties} and $\widetilde{a}_\ell=\sum\limits_{k_1,\ldots,k_\ell}c^{k_1\ldots k_\ell}b_{k_1}\otimes\ldots\otimes b_{k_\ell}$, the sum over $k_1,\ldots, k_\ell$ being finite.
\end{definition}

\begin{remark}
	\noindent
	\begin{enumerate}[(i)]
		\item
			If $\ell=M$ one has $\iota_M^{j_1\ldots j_M}(\widetilde{a}_M)=\widetilde{a}_M$.
			Moreover, by direct inspection $\iota_M^{j_1\ldots j_{\ell-1}}$ does not depend on the chosen basis $I,b_1,\ldots,b_{\kappa^2-1}$.
			%	Indeed if $\{\overline{b}_j\}_{j=1}^{\kappa^2-1}$ is another basis we have $b_j=Q_j^k\overline{b}_k$ for an invertible matrix $Q$ and thus
			%	\begin{align*}
				%		\widetilde{a}_\ell
				%		=c^{k_1\ldots k_\ell}
				%		b_{k_1}\otimes\ldots\otimes\otimes b_{k_\ell}
				%		&=c^{k_1\ldots k_\ell}Q^{p_1}_{k_1}\cdots Q^{p_\ell}_{k_\ell}
				%		\overline{b}_{p_1}\otimes\ldots\otimes\overline{b}_{p_\ell}
				%		\\&=\overline{c}^{p_1\ldots p_\ell}
				%		\overline{b}_{p_1}\otimes\ldots\otimes\overline{b}_{p_\ell}\,.
				%	\end{align*}
			%	We thus have
			%	\begin{align*}
				%		\overline{c}^{p_1\ldots p_\ell}
				%		\overline{b}_{p_1}\otimes I^{j_1}\otimes\ldots\otimes I^{j_{\ell-1}}\otimes \overline{b}_{p_\ell}
				%		&=\overline{c}^{k_1\ldots k_\ell}Q^{p_1}_{k_1}\cdots Q^{p_\ell}_{k_\ell}
				%		\overline{b}_{p_1}\otimes I^{j_1}\otimes\ldots\otimes I^{j_{\ell-1}}\otimes \overline{b}_{p_\ell}
				%		\\&=c^{k_1\ldots k_\ell}
				%		b_{k_1}\otimes I^{j_1}\otimes\ldots\otimes I^{j_{\ell-1}}\otimes b_{k_\ell}\,.
				%	\end{align*}
			\item
			On account of Definition \ref{Def: Btilde-irreducible element} we have $\iota_M^{j_1\ldots j_{\ell-1}}(\widetilde{B}^\ell)\subseteq B^M_{\textsc{irr}}$.
			Moreover, $\iota_M^{j_1\ldots j_{\ell-1}}$ is injective.
			Indeed, if $\iota_M^{j_1\ldots j_{\ell-1}}(\widetilde{a}_\ell)=0$ then
			for all $\eta_1,\ldots,\eta_\ell\in S(B)$ we have
			\begin{align*}
				0=[\eta_1\otimes \tau^{j_1}\otimes\ldots\otimes \tau^{j_{\ell-1}}\otimes\eta_\ell](\iota_M^{j_1\ldots j_{\ell-1}}(\widetilde{a}_\ell))
				=(\eta_1\otimes\ldots\otimes\eta_\ell)(\widetilde{a}_\ell)\,.
			\end{align*}
			The arbitrariness of $\eta_1,\ldots,\eta_\ell$ entails $\omega_\ell(\widetilde{a}_\ell)=0$ for all $\omega_\ell\in S(B^\ell)$, therefore, $\widetilde{a}_\ell=0$.
	\end{enumerate}
\end{remark}

Let $I,b_1,\ldots,b_{\kappa^2-1}$ be a basis of $B$ fulfilling \eqref{Eq: Mk-basis properties} and let $a_M\in B^M_{\textsc{irr}}$, $M\geq 2$.
By considering Equation \eqref{Eq: decomposition of BM into part with a different number of identities} for $a_M\in B^M_{\textsc{irr}}$ we find
\begin{multline}\label{Eq: Btilde irreducible element decomposition}
	a_M=\iota_M^{M-2}(\widetilde{a}_2)
	+\sum_{j_1+j_2=M-3}\iota_M^{j_1j_2}(\widetilde{a}_{3|j_1j_2})
	\\+\ldots
	+\sum_{j_1+\ldots+j_{\ell-1}=M-\ell}
	\iota_M^{j_1\ldots j_{\ell-1}}(\widetilde{a}_{\ell|j_1\ldots j_{\ell-1}})
	+\ldots
	+\widetilde{a}_M\,,
\end{multline}
where $\widetilde{a}_{\ell|j_1\ldots j_{\ell-1}}\in \widetilde{B}^\ell$ for all $\ell,j_1,\ldots,j_{\ell-1}$.

Equation \eqref{Eq: Btilde irreducible element decomposition} provides a description of $B^M_{\textsc{irr}}$ in terms of "$\widetilde{B}$-components".
To this avail we consider the vector space
\begin{align}
	\boldsymbol{\widetilde{B}}^M&:=\widetilde{B}^2
	\oplus\bigotimes_{j_1+j_2=M-3}\widetilde{B}^3\oplus\ldots
	\oplus\bigotimes_{j_1+\ldots+j_{\ell-1}=M-\ell}\widetilde{B}^\ell\oplus\ldots
	\oplus\widetilde{B}^M\,,
\end{align}
where $\boldsymbol{\widetilde{B}}^0=\mathbb{C}$ and $\boldsymbol{\widetilde{B}}^1=\widetilde{B}$.
We then define the linear map
\begin{multline}\label{Eq: Btilde irreducible element decomposition map}
	\Phi_M\colon\boldsymbol{\widetilde{B}}^M\to B^M\;
	\Phi_M(\boldsymbol{\widetilde{a}}_M):=
	\begin{dcases}
		a_0\quad M=0\\
		\widetilde{a}_1\quad M=1\\
		\iota_M^{M-2}(\widetilde{a}_2)+\ldots
		\\+\sum_{j_1+\ldots+j_{\ell-1}=M-\ell}
		\iota_M^{j_1\ldots j_{\ell-1}}(\widetilde{a}_{\ell|j_1\ldots j_{\ell-1}})
		+\ldots
		+\widetilde{a}_M\quad M\geq 2
	\end{dcases}\,,
\end{multline}
where $\boldsymbol{\widetilde{a}}_M\in\boldsymbol{\widetilde{B}}^M$ is given by
\begin{align}\label{Eq: BtildeM generic element}
	\boldsymbol{\widetilde{a}}_M
	=\begin{dcases*}
		a_0\quad M=0\\
		\widetilde{a}_1\quad M=1\\
		\widetilde{a}_2\oplus
		\bigotimes_{j_1+j_2=M-3}\widetilde{a}_{3|j_1j_2}\oplus\ldots
		\oplus\bigotimes_{j_1+\ldots +j_{\ell-1}=M-\ell}\widetilde{a}_{\ell|j_1+\ldots+j_{\ell-1}}\oplus\ldots
		\oplus\widetilde{a}_M
		\quad M\geq 2 
	\end{dcases*}\,.
\end{align}
Equation \eqref{Eq: Btilde irreducible element decomposition} can be rephrased by saying that for all $a_M\in B^M_{\textsc{irr}}$ there exists $\boldsymbol{\widetilde{a}}_M\in\boldsymbol{\widetilde{B}}^M$ such that $\Phi_M(\boldsymbol{\widetilde{a}}_M)=a_M$.
The following lemma shows that $\Phi_M$ is in fact an isomorphism, proving that $B^M_{\textsc{irr}}\simeq\boldsymbol{\widetilde{B}}^M$.

\begin{lemma}\label{Lem: PhiM is a bijection}
	For all $M\in\mathbb{N}$, the map $\Phi_M\colon\boldsymbol{\widetilde{B}}^M\to B^M_{\textsc{irr}}$ is an isomorphism.
\end{lemma}
\begin{proof}
	There is nothing to prove for $M\in\{0,1\}$, therefore, we assume $M\geq 2$.
	
	From equations \eqref{Eq: Btilde irreducible element decomposition}, \eqref{Eq: Btilde irreducible element decomposition map}, we have that $\Phi_M$ is linear and surjective: Thus, it remains to prove that $\Phi(\boldsymbol{\widetilde{a}}_M)=0$ implies $\boldsymbol{\widetilde{a}}_M=0$.
	We now prove that, if $\Phi(\boldsymbol{\widetilde{a}}_M)=0$, then all components of $\boldsymbol{\widetilde{a}}_M$ appearing in Equation \eqref{Eq: BtildeM generic element} vanish.
	
	To this avail let $\eta_1,\eta_2\in S(B)$.
	By direct inspection we have
	\begin{align*}
		0=(\eta_1\otimes\tau^{M-2}\otimes\eta_2)[\Phi(\boldsymbol{\widetilde{a}}_M)]
		=(\eta_1\otimes\eta_2)(\widetilde{a}_2)\,.
	\end{align*}
	Notice that no other term from $\Phi(\boldsymbol{\widetilde{a}}_M)$ provides a non-vanishing contribution because $\tau(\widetilde{B})=\{0\}$.
	The arbitrariness of $\eta_1,\eta_2\in S(B)$ leads to $\omega_2(\widetilde{a}_2)=0$ for all $\omega_2\in S(B^2)$ and thus $\widetilde{a}_2=0$.
	
	We now proceed by proving that $\widetilde{a}_{3|j_1 j_2}=0$ for all $j_1+j_2=M-3$.
	To this avail let $j_1,j_2$ be such that $j_1+j_2=M-3$ and let $\eta_1,\eta_2,\eta_3\in S(B)$.
	Since we already proved that $\widetilde{a}_2=0$ it follows that
	\begin{align*}
		0=(\eta_1\otimes\tau^{j_1}\otimes\eta_2\otimes \tau^{j_2}\otimes\eta_3)[\Phi(\boldsymbol{\widetilde{a}}_M)]
		=(\eta_1\otimes\eta_2\otimes\eta_3)(\widetilde{a}_{3|j_1j_2})\,.
	\end{align*}
	Once again, the arbitrariness of $\eta_1,\eta_2,\eta_3\in S(B)$ (as well as the one of $j_1,j_2$) leads to $\widetilde{a}_{3|j_1j_2}=0$ for all $j_1+j_2=M-3$.
	Proceeding by induction we find $\widetilde{a}_{\ell|j_1\ldots j_{\ell-1}}=0$ for all $j_1+\ldots+j_{\ell-1}=M-\ell$.
	Thus $\boldsymbol{\widetilde{a}}_M=0$.
\end{proof}


\begin{thebibliography}{ll}
    \bibitem{Bayen_Flato_Fronsdal_Lichnerowicz_1978_I}
    Bayen F., Flato M., Fronsdal C., Lichnerowicz A., Sternheimer D.,
    \textit{Deformation Theory and Quantization. 1.Deformations of Symplectic Structures}, 
    Annals Phys. 111 (1978) 61
    \href{https://doi.org/10.1016/0003-4916(78)90224-5}
    {(1978)}.
    
    \bibitem{Bayen_Flato_Fronsdal_Lichnerowicz_1978_II}
    Bayen F., Flato M., Fronsdal C., Lichnerowicz A., Sternheimer D.,
    \textit{Deformation Theory and Quantization. 2. Physical Applications},
    Annals Phys. 111 61 
    \href{https://doi.org/10.1016/0003-4916(78)90225-7}
    {(1978)}.

    \bibitem{Berezin_1975}
    Berezin F.A.,
    \textit{General concept of quantization},
    Commun. Math. Phys. 40, 153-174
    \href{https://doi.org/10.1007/BF01609397}
    {(1975).}
    
	\bibitem{Blackadar_2006}
	Blackadar B.,
	\textit{Operator algebras -Theory of $C^*$-algebras and von Neumann algebras},
	Springer-Verlag Berlin Heidelberg
	\href{https://doi.org/10.1007/3-540-28517-2}
	{(2006)}.
	
	\bibitem{Blackadar_Kirchberg_1997}
	Blackadar B., Kirchberg E.,
	\textit{Generalized inductive limits of finite-dimensional $C^*$-algebras},
	Math. Ann. 307, 343-380
	\href{https://doi.org/10.1007/s002080050039}
	{(1997)}.
	
	\bibitem{Bratteli_Robinson_1997}
	Bratteli O., Robinson D.W.,
	\textit{Operator algebras and quantum statistical mechanics II}
	Springer-Verlag Berlin Heidelberg 
	\href{https://doi.org/10.1007/978-3-662-09089-3}
	{(1981,1997)}.
	
	\bibitem{Dixmier_1977}
	Dixmier J.,
	\textit{$C^*$-Algebras},
	North-Holland, (1977).

        \bibitem{Drago_vandeVen_2023}
        Drago, N., van de Ven, C.J.F.,
        \textit{DLR-KMS correspondence on lattice spin systems},
	Letters in Mathematical Physics 113, 88,
	\href{https://doi.org/10.1007/s11005-023-01710-x}
        {(2023)}.

        \bibitem{Friedli_Velenik_2017}
	Friedli S., Velenik Y.,
	\textit{Statistical mechanics of lattice systems - A concrete mathematical introduction},
	Cambridge University Press
	\href{https://doi.org/10.1017/9781316882603}
	{(2017)}.
        
	\bibitem{Landsman_1998}
	Landsman K.,
	\textit{Mathematical Topics Between Classical and Quantum Mechanics},
	Springer New York, NY 
	\href{https://doi.org/10.1007/978-1-4612-1680-3}
	{(1998)}.
	
	\bibitem{Landsman_2017}
	Landsman K.,
	\textit{Foundations of Quantum Theory: from classical concepts to operators algebras},
	Springer Cham 
	\href{https://doi.org/10.1007/978-3-319-51777-3}
	{(2017)}.

	\bibitem{Landsman_Moretti_vandeVen_2020}
	Landsman K., Moretti V., van de Ven C. J. F.,
	\textit{Strict deformation quantization of the state space of $M_k(\mathbb{C})$ with application to the Curie-Weiss model},
	Rev. Math. Phys. Vol. 32, No. 10, 2050031 
	\href{https://doi.org/10.1142/S0129055X20500312}
	{(2020)}.
		
	\bibitem{Moretti_vandeVen_2020}
	Moretti V., van de Ven C. J. F.,
	\textit{Bulk-boundary asymptotic equivalence of two strict deformation quantizations},
	Letters in Mathematical Physics 110(11), 2941-2963
	\href{https://doi.org/10.1007/s11005-020-01333-6}
	{(2020)}.
	
	\bibitem{Moretti_vandeVen_2021}
	Moretti V., van de Ven C. J. F.,
	\textit{The classical limit of Schr\"odinger operators in the framework of Berezin quantization and spontaneous symmetry breaking as emergent phenomenon},
	International Journal of Geometric Methods in Modern Physics
	\href{https://doi.org/10.1142/S0219887822500037}
	{(2022)}.
	
	\bibitem{Murro_vandeVen_2022}
	Murro S., van de Ven  C. J. F.
	\textit{Injective tensor products in strict deformation quantization},
	Math Phys Anal Geom 25, 2 
	\href{https://doi.org/10.1007/s11040-021-09414-1}
	{(2022)}.
	
	\bibitem{Raggio_Werner_1989}
	Raggio G. A., Werner R. F.,
	\textit{Quantum statistical mechanics of general mean field systems},
	Helv. Phys. Acta 62 980
	\href{http://dx.doi.org/10.5169/seals-116175}
	{(1989)}.
	
	\bibitem{Rieffel_1994}
	Rieffel M. A.,
	\textit{Quantization and $C^*$-algebras},
	Contemp. Math. 167 67-97
	\href{http://www.ams.org/books/conm/167/conm167-endmatter.pdf}
	{(1994)}.

	\bibitem{vandeVen_2020}
	van de Ven C. J. F.,	
	\textit{The classical limit of mean-field quantum spin systems},
	J. Math. Phys. 61, 121901 
	\href{https://doi.org/10.1063/5.0021120}
	{(2020)}.
	
	\bibitem{vandeVen_2021}
	van de Ven C. J. F.,
	\textit{The classical limit and spontaneous symmetry breaking in algebraic quantum theory},
	Expo. Math. 40, 3 
	\href{https://doi.org/10.48550/arXiv.2109.05653}
        {(2022)}.
\end{thebibliography}
\end{document}